\documentclass[runningheads,a4paper,%
  orivec,envcountsame,envcountsect]{llncs}
\usepackage[utf8]{inputenc}
\usepackage[T1]{fontenc}
\usepackage[british]{babel}
\usepackage{xspace}

\usepackage[%
rm={oldstyle=false,proportional=true},%
sf={oldstyle=false,proportional=true},%
tt={oldstyle=false,proportional=true,variable=true},%
qt=false%
]{cfr-lm}

\usepackage[noadjust]{cite} 
\usepackage{graphicx,paralist,cite,csquotes}

\usepackage{amssymb,mathrsfs,mathtools,bm}
\numberwithin{equation}{section}
\usepackage[all,2cell]{xy}
\UseAllTwocells
\SelectTips{cm}{}
\newdir{ >}{{}*!/-5pt/@{>}}
\newdir{ (}{{}*!/-10pt/\dir^{(}}
\newdir{ )}{{}*!/-10pt/\dir_{(}}

\let\doendproof\endproof
\renewcommand\endproof{~\hfill\qed\doendproof}
\usepackage[colorlinks=false]{hyperref}

\usepackage[capitalise]{cleveref}

\usepackage{microtype}

\spnewtheorem{assumption}[theorem]{Assumption}{\bfseries}{\rmfamily}
\spnewtheorem{notation}[theorem]{Notation}{\bfseries}{\rmfamily}
\spnewtheorem{observation}[theorem]{Observation}{\bfseries}{\rmfamily}
\spnewtheorem{defn}[theorem]{Definition}{\bfseries}{\rmfamily}
\spnewtheorem{expl}[theorem]{Example}{\bfseries}{\rmfamily}
\spnewtheorem{rem}[theorem]{Remark}{\bfseries}{\rmfamily}

\Crefname{notation}{Notation}{Notations}
\Crefname{rem}{Remark}{Remarks}
\Crefname{expl}{Example}{Examples}
\Crefname{defn}{Definition}{Definitions}
\Crefname{assumption}{Assumption}{Assumptions}

\newcommand{\Stone}{\mathbf{Stone}}

\newcommand{\Priest}{\mathbf{Priest}}

\newcommand{\Set}{\mathbf{Set}}
\newcommand{\Pos}{\mathbf{Pos}}
\newcommand{\Top}{\mathbf{Top}}
\newcommand{\JSL}{\mathbf{JSL}}

\newcommand{\Alg}{\mathsf{Alg}}
\newcommand{\A}{\mathscr{A}}
\newcommand{\B}{\mathscr{B}}
\newcommand{\C}{\mathscr{C}}
\newcommand{\D}{\mathscr{D}}
\newcommand{\X}{\mathscr{X}}
\newcommand{\E}{\mathcal{E}}
\newcommand{\V}{\mathcal{V}}
\newcommand{\MT}{\mathbf{T}}

\newcommand{\Pro}[1]{\operatorname{Pro-#1}}

\newcommand{\hatD}{\widehat\D}
\newcommand{\hatT}{\widehat\MT}
\newcommand{\hatt}{\hat{T}}

\newcommand{\br}[1]{{#1^+}}
\newcommand{\brv}[1]{{#1^+_\V}}

\newcommand{\hatJ}{\hat{J}}

\newcommand{\hatX}{{\hat{X}}}
\newcommand{\hateta}{\hat\eta}
\newcommand{\hatmu}{\hat\mu}
\renewcommand{\epsilon}{\varepsilon}
\newcommand{\id}{\mathit{id}}
\newcommand{\Id}{\mathsf{Id}}
\newcommand{\seq}{\subseteq}
\newcommand{\xra}{\xrightarrow}
\DeclareMathOperator{\Lim}{Lim}

\newcommand{\defeq}{\coloneqq}

\newcommand{\op}{\mathrm{op}}

\renewcommand{\o}{\cdot}

\newcommand{\takeout}[1]{\empty}

\renewcommand{\phi}{\varphi}
\newcommand{\ra}{\rightarrow}

\newcommand{\Ra}{\Rightarrow}

\newcommand{\DCat}{\mathscr{D}}

\newcommand{\epito}{\twoheadrightarrow}

\newcommand{\monoto}{\rightarrowtail}

\newcommand{\Q}{\mathcal{Q}}

\newcommand{\Pow}{\mathcal{P}}

\title{Profinite Monads, Profinite Equations, and Reiterman's Theorem}
\titlerunning{Profinite Monads, Profinite Equations, and Reiterman's Theorem}

\author{Liang-Ting Chen\inst{1}\fnmsep\thanks{Liang-Ting Chen acknowledges
    gratefully partial support from AFOSR.} \and Ji\v{r}í Ad\'amek\inst{1} \and
  Stefan
  Milius\inst{2}\fnmsep\thanks{Stefan Milius acknowledges support by the
    Deutsche Forschungsgemeinschaft (DFG) under project MI 717/5-1} \and Henning
  Urbat\inst{1}\fnmsep\thanks{Ji\v{r}í Ad\'amek and Henning Urbat acknowledge
    support by the Deutsche Forschungsgemeinschaft (DFG) under project AD
    187/2-1}}
\authorrunning{L.-T.~Chen, J.~Adámek, S.~Milius, H.~Urbat}

\institute{Institut f\"{u}r Theoretische Informatik \\
Technische Universit\"at Braunschweig, Germany
\and
Lehrstuhl f\"ur Theoretische Informatik \\
Friedrich-Alexander Universit\"at Erlangen-N\"urnberg, Germany
}

\begin{document}
\maketitle

\begin{abstract}
  Profinite equations are an indispensable tool for the algebraic 
  classification of
  formal languages. Reiterman's theorem states that they precisely specify
  pseudovarieties, i.e.  classes of finite algebras closed under finite 
  products, subalgebras and quotients. In this paper Reiterman's theorem is 
  generalised to finite
  Eilenberg-Moore algebras for a monad~$\MT$ on a variety~$\D$ of (ordered)
  algebras: a class of finite $\MT$-algebras is a pseudovariety iff it is 
  presentable by profinite
  (in-)equations. As an application, quasivarieties of finite algebras are shown
  to be presentable by profinite implications. Other examples include 
  finite ordered algebras, finite categories, finite $\infty$-monoids, etc.
\end{abstract}

\section{Introduction}

Algebraic automata theory investigates the relationship between the behaviour 
of finite machines and descriptions of these behaviours in terms of finite 
algebraic 
structures. For example, regular languages of finite words are precisely the 
languages recognised by finite monoids. And Sch\"utzenberger's 
theorem~\cite{sch65} shows that star-free regular languages 
correspond to aperiodic finite monoids, which easily leads to the 
decidability of star-freeness. A generic correspondence result of this kind is 
Eilenberg's variety theorem~\cite{Eilenberg1976}. It gives a bijective correspondence 
between \emph{varieties of languages} (classes of regular languages closed under 
boolean operations,  derivatives and homomorphic preimages) and
\emph{pseudovarieties of monoids} (classes of finite monoids closed under 
finite products, submonoids and quotients). Another, more syntactic, 
characterisation of pseudovarieties follows from Reiterman's theorem~\cite{Reiterman1982} (see also Banaschewski~\cite{Banaschewski1983}): 
they are precisely the classes of finite monoids 
specified by \emph{profinite equations}.

In the meantime Eilenberg-type correspondences have been discovered for other 
kinds of algebraic structures, including ordered monoids~\cite{Pin1995}, 
idempotent semirings~\cite{polak01}, associative algebras over a 
field~\cite{reutenauer80} and Wilke algebras~\cite{wilke91}, always with 
rather similar proofs. This has spurred recent interest in 
 generic approaches to algebraic language theory that can produce such 
correspondences as instances of a single 
result. Boja\'nczyk~\cite{Bojanczyk2015} extends the classical notion of 
language recognition by monoids (viewed as algebraic structures over the 
category of 
sets) to algebras for an arbitrary monad on many-sorted sets. He also presents 
an 
Eilenberg-type theorem at this level of generality, interpreting a result of 
Almeida \cite{almeida90} in categorical terms. Our previous work 
in~\cite{Adamek2014c,Adamek2015,Adamek2015b, Chen2015} takes an orthogonal 
approach: one 
keeps monoids but considers them in categories $\D$ of (ordered) algebras such as 
posets, semilattices and vector spaces. 
Analysing the latter work it 
becomes clear that the step from sets to more general 
categories $\D$ is necessary to obtain the right notion of language recognition by 
finite monoids; e.g.~to cover Pol\'ak's Eilenberg-type theorem for idempotent 
semirings \cite{polak01}, one needs to take the base category $\D$ of 
semilattices. On the other hand, from Boja\'nczyk's work it is clear that one 
also has 
to generalise from monoids to other algebraic structures if one wants to 
capture such 
examples as Wilke algebras.

 The present paper is the first step in a line of work that considers a common 
 roof for both approaches, working with algebras for a monad $\MT$ on an 
 arbitrary variety $\D$ of many-sorted, possibly ordered algebras.
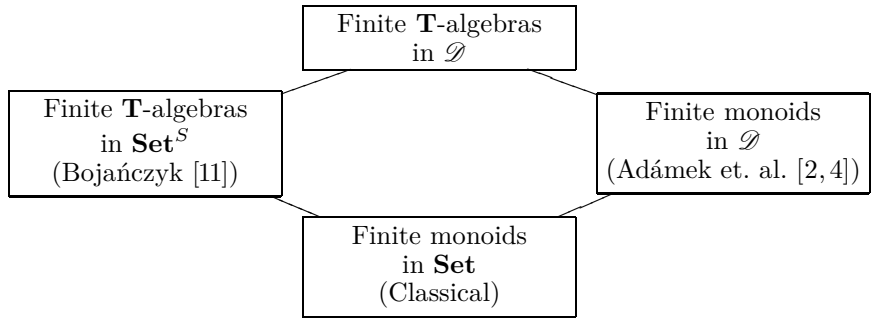
\[
\xymatrix@=8pt{
& *+[F]{\txt<8pc>{Finite $\MT$-algebras\\ in $\D$}} \ar@{-}[dl] \ar@{-}[dr] &\\
*+[F]{\txt<8pc>{Finite $\MT$-algebras \\ in $\Set^S$ \\ (Boja\'nczyk \cite{Bojanczyk2015})}} \ar@{-}[dr] & & *+[F]{\txt<8pc>{Finite monoids\\ in $\D$  \\ (Ad\'amek et. al. \cite{Adamek2014c,Adamek2015})}} \ar@{-}[dl]\\
& *+[F]{\txt<8pc>{Finite monoids\\ in $\Set$ \\ (Classical)}} &
}
\] Our main 
contribution is a generalisation of Reiterman's 
theorem, stating that pseudovarieties of finite algebras 
are presentable by profinite equations, to the more 
general situation of algebras for a monad. Starting with a variety $\D$, we 
form the pro-completion of the full subcategory $\D_f$ of finite 
algebras, 
\[\hatD \defeq \Pro \D_f.\] For example, for $\D =$ sets, posets and monoids 
we get $\hatD$ = Stone spaces, Priestley spaces and profinite monoids. 
Next, we consider a monad $\MT$ on $\D$ and associate to it a monad~$\hatT$ 
on~$\hatD$, called the
\emph{profinite monad} of~$\MT$. For example, if $\D = \Set$ and $\MT$ is
the finite word monad (whose algebras are precisely monoids), then $\hatT$
is the monad of profinite monoids on the category of Stone spaces; that is, 
$\hatT$ associates to each finite Stone space (= finite set) $X$ the space 
$\widehat{X^*}$ of profinite words on $X$. Similarly, for the monad~$\MT$ of 
finite and infinite words on $\Set$ (whose algebras we call 
\emph{$\infty$-monoids}) the profinite monad $\hatT$ constructs the space of 
\emph{profinite $\infty$-words}.

The classical profinite equations for monoids, used for presenting
pseudovarieties of monoids, are generalised to \emph{profinite equations}~$u = v$ that are pairs
of elements of~$\hatT \hat{\Phi_X}$, where $\hat{\Phi_X}$ is the free 
profinite 
$\D$-algebra on a finite set $X$ of variables. Our main result is that profinite equations present precisely classes of finite $\MT$-algebras
closed under finite products, subalgebras, and quotients.

We will additionally study a somewhat unusual 
concept of profinite equation where in lieu of finite \emph{sets} $X$ of 
variables we 
use finite \emph{algebras} $X\in\D_f$ of variables. The classes 
of
finite $\MT$-algebras presented by such profinite equations are then
precisely those closed under finite products, subalgebras, and split
quotients. These two variants are actually instances of a general result 
(\Cref{thm:reiterman}) that is parametric in a class $\X$ of ``algebras 
of variables'' in $\D$. 

The above results hold if~$\D$ is a variety of algebras. In case that 
$\D$ is a variety of ordered algebras, we obtain the analogous two results, working
with \emph{profinite inequations}~$u\leq v$ instead of equations.  As 
instances we recover Reiterman's original
 theorem~\cite{Reiterman1982} and its version for ordered algebras due to Pin
 and Weil~\cite{PinWeil1996}. Another consequence of our theorem is the
 observation that \emph{quasivarietes of finite algebras in $\D$}, i.e.\ subclasses of $\D_f$ closed under finite products and subalgebras, are presentable by \emph{profinite implications}. Moreover, we obtain a number of new Reiterman-type results. For example, for the monad of finite and infinite words on $\Set$, our Reiterman theorem shows that a class of finite $\infty$-monoids is a pseudovariety iff it can be
presented by equations between profinite $\infty$-words.  Finally, we can also 
treat 
categories of 
$\MT$-algebras that are not varieties. E.g.~by taking for $\D$ the category of 
graphs and $\MT$ the free-category monad we essentially recover a result of 
Jones on pseudovarieties of finite categories~\cite{Jones1996}.

\section{Preliminaries} \label{sec:preliminaries}
In this section we review the necessary concepts from category theory,  
universal algebra and topology we will use throughout the paper. Recall that 
for a finitary many-sorted signature~$\Gamma$ a \emph{variety of
  $\Gamma$-algebras} is a full subcategory of $\Alg_\Gamma$, the category of
$\Gamma$-algebras, specified by equations $s = t$ between $\Gamma$-terms.  By
Birkhoff's HSP theorem varieties are precisely the classes of algebras closed
under products, subalgebras, and quotients (= homomorphic images). Similarly,
\emph{ordered $\Gamma$-algebras} are posets equipped with order-preserving
$\Gamma$-operations, and their morphisms are order-preserving 
$\Gamma$-homomorphisms. A
\emph{quotient} of an ordered algebra $B$ is represented by a surjective morphism 
$e\colon
B \epito A$, and a \emph{subalgebra} of~$B$ is represented by an 
\emph{order-reflecting}
morphism~$m\colon A
\rightarrowtail B$, i.e.\ $mx \leq
my$ iff $x \leq y$.  A \emph{variety of ordered $\Gamma$-algebras} is a full
subcategory of $\Alg_{\leq \Gamma}$, the category of ordered $\Gamma$-algebras,
specified by inequations $s \leq t$ between $\Gamma$-terms. By Bloom's HSP
theorem~\cite{Bloom1976}, varieties of ordered algebras are precisely the
classes of ordered algebras closed under products, subalgebras and quotients.

\begin{rem}
For notational simplicity we restrict our attention to single-sorted 
varieties. However, all definitions, theorems and proofs
that follow are easily adapted to a many-sorted setting. See also
\Cref{rem:many-sorted} and \Cref{ex:jones}.
\end{rem}

\begin{defn}
Let $\D$ be a variety  of algebras or ordered algebras.
\begin{enumerate}[(a)]
\item  A \emph{topological $\D$-algebra} is a topological space endowed with a
    $\D$-algebraic structure such that every operation is continuous with
    respect to the product topology. \emph{Morphisms} of topological 
    $\D$-algebras are continuous $\D$-morphisms. 
\item A topological $\D$-algebra is \emph{profinite} if it is a cofiltered 
limit of finite $\D$-algebras with discrete topology.
\end{enumerate}
\end{defn}

\begin{notation}
Throughout this paper we fix a variety $\D$ of algebras or ordered algebras, equipped with the factorisation system of quotients and subalgebras. We denote by $\hatD$ the category of profinite $\D$-algebras. We use the forgetful functors
\[ \xymatrix{  \D_f \ar@{ >->}[r]^\hatJ  &\hatD \ar[r]^V & \D} \] where $V$ forgets the topology and 
$\hatJ$ views a finite $\D$-algebra as a profinite $\D$-algebra with discrete 
topology. We will often identify $A\in \D_f$ with $\hatJ A$.
\end{notation}

\begin{expl}
  \begin{enumerate}
    \item $\widehat{\Set}$ is the category $\Stone$ of Stone spaces, i.e.\
      compact spaces such that any two distinct elements can be 
      separated by a clopen set.
     \item Let $\Pos$ be the category of posets and monotone maps, viewed as the
       variety of ordered algebras over the empty signature. Then
       $\widehat{\Pos}$ is the category $\Priest$ of Priestley
       spaces~\cite{Priestley1972}, i.e.\ ordered compact spaces such that for
       any two elements $u, v$ with $u \not\leq v$ there is a clopen upper 
       set containing $u$ but not $v$.

    \item For the variety $\mathbf{Mon}$ of monoids, the category 
    $\widehat{\mathbf{Mon}}$ consists  of all monoids in 
    $\Stone$;
      that is, a topological monoid is profinite iff it carries a Stone
      topology. Analogous descriptions of $\hatD$ hold for most familiar 
      varieties $\D$ over a finite signature, e.g.
      groups, semilattices, vector spaces over a finite field;
      see~\cite{Johnstone1982}.  
  \end{enumerate}
\end{expl}

\begin{rem}\label{re:hatD-is-procompletion}
 By~\cite[Remark 
VI.2.4]{Johnstone1982} the category $\hatD$ is the \emph{pro-completion}, i.e.\ 
the free completion under cofiltered limits, of $\D_f$. Hence $\hatD$ is dual to a \emph{locally finitely presentable} category~\cite{Adamek1994}, which entails the following properties:
\begin{enumerate}[(i)]
  \item Every object $A$ of $\hatD$ is the cofiltered limit of all morphisms 
  $h\colon A\to A'$ with finite codomain. More precisely, if $(A\downarrow
  \D_f)$ denotes the comma category of all such morphisms $h$, the diagram
  \[
    (A\downarrow \D_f) \to \hatD, \quad h\mapsto A',
  \] 
  has the limit $A$ with limit projections $h$.  \item Given a cofiltered limit
  cone $\pi_i\colon A\ra A_i$ ($i\in I$) in $\hatD$, any morphism $f\colon A\to
  B$ with finite codomain factors through some $\pi_i$.
\end{enumerate}
\end{rem}

\begin{lemma}
  \label{lem:hatD-factorisation}
  $\hatD$ has the factorisation system of surjective morphisms and injective (resp. order-reflecting) morphisms.  
\end{lemma}

\begin{defn}
The \emph{profinite completion} of an object $D\in\D$ is the limit $\hat D\in 
\hatD$ of the cofiltered diagram 
\[
  (D\downarrow \D_f) \to \hatD, \quad (h\colon D \to D') \mapsto D'
\] 
We denote the limit projection corresponding to $h: D\to D'$ by $\hat h\colon \hat D \to D'$. Observe that
$\hat D = D$ for any $D\in \D_f$, and $\widehat h = h$ for any  
morphism 
$h$ in $\D_f$.
\end{defn}

\begin{proposition}\label{prop:profcomp}
 The maps $D\mapsto \hat D$ and $h\mapsto \hat h$ extend to a left adjoint for 
 the forgetful functor $V$, denoted by \[\widehat{\,\cdot\,}\colon \D\to\hatD.\]
\end{proposition}

\begin{rem}\label{rem:homtheorem}
We will frequently use the following facts:
\begin{enumerate}[(a)]
\item    \emph{Homomorphism theorem.} Given morphisms 
    $e\colon A \epito B$ and $f\colon A \to 
        C$ in $\D$ with $e$ surjective, there exists a morphism $f'$ with $f'\o e = f$ iff 
    $e(a)=e(a')$ 
    implies $f(a)=f(a')$ (resp.\ $e(a)\leq e(a')$ implies $f(a)\leq f(a')$) for all 
    $a,a'\in A$. Moreover, if $A,B,C$ are topological $\D$-algebras with a compact 
      Hausdorff topology and $e$ and $f$ are continuous $\D$-morphisms, then 
      $f'$ is continuous.
\item The forgetful functor $\left|\mathord{-}\right|\colon \D\to\Set$ has a left 
adjoint assigning to each set $X$ the free $\D$-algebra $\Phi_X$ on $X$. 
\item Free $\D$-algebras are projective: for any morphism $f\colon \Phi_X \to B$ 
and any surjective morphism  $e\colon A\epito B$ in $\D$ there exists a morphism 
$f'\colon \Phi_X\to A$ with $e\o f' = f$. Indeed, choose a function $m\colon 
\left|B\right| \to
    \left|A\right|$ with $e\o m = \id$. Then the restriction of $m\o f$ to $X$
    extends to a morphism $f'\colon \Phi_X\to A$ of $\D$ with $f=f'\o e$, since the
    morphisms on both sides agree on the generators $X$.
\end{enumerate} 
\end{rem}

\begin{notation}
  For a monad $\MT = (T, \eta, \mu)$ on $\D$, we write $\D^\MT$ for the category of
  $\MT$-algebras and $\MT$-homomorphisms, and $\D_f^\MT$ for the 
  full subcategory of finite $\MT$-algebras. The forgetful functors are denoted by
  \[ U\colon \D^\MT_f \to \D_f\quad\text{and}\quad U^\MT\colon \D^\MT 
  \to 
  \D.
  \] Recall that $U^\MT$ has a left adjoint
  mapping~$D\in\D$ to its free $\MT$-algebra~$(TD, \mu_D)$.
\end{notation}

\begin{rem}
  \label{re:factorisation_DT}
 If $T$ preserves surjective morphisms, the homomorphism theorem applies 
 to 
   $\MT$-algebras. That is, if $A,B,C$ in 
   \Cref{rem:homtheorem}(a) are $\MT$-algebras and $e$ and $f$ are 
   $\MT$-homomorphisms, so is $f'$.  Moreover the 
   factorisation system of $\D$ lifts
  to~$\D^\MT$: every $\MT$-homomorphism $h\colon (A,\alpha)\to (B,\beta)$ can be 
  factorised into a 
  surjective $\MT$-homomorphism followed by an injective (resp. 
  order-reflecting) one. 
 \emph{Quotients} and \emph{subalgebras} of $\MT$-algebras are 
  taken w.r.t. this factorisation system.
\end{rem}

\begin{expl}\label{ex:monads}
We are mainly interested in monads representing structures in algebraic 
language theory.
\begin{enumerate}[(a)]
\item \emph{Finite words.} The classical example is the free-monoid monad 
$\MT$ on $\D=\Set$,
\[ TX = X^* = \coprod_{n<\omega} X^n.\]
The importance of the monad 
$\MT$ is that functions $TX\to \{0,1\}$ correspond to languages of finite 
words over the alphabet $X$, and regular languages are precisely the 
languages recognized by finite $\MT$-algebras (= finite monoids).
Boja\'nczyk~\cite{Bojanczyk2015} recently gave a generalisation of the
classical 
Eilenberg theorem to arbitrary monads $\MT$ on $\Set$, relating 
pseudovarieties of 
finite $\MT$-algebras to varieties of 
$\MT$-recognisable languages.
\item \emph{Finite words over semilattices.} From the perspective of algebraic 
language theory it is natural to study monoids in algebraic categories beyond 
$\Set$. For example, let $\D=\JSL$ be the variety of 
join-semilattices with $0$, considered as a monoidal category 
w.r.t. the usual 
tensor 
product. The free-monoid monad on $\JSL$ is given by
\[ TX = X^\circledast = \coprod_{n<\omega} X^{\otimes n}, \]
the coproduct of all finite tensor powers of $X$, and $\MT$-algebras are 
precisely idempotent semirings. In case $X=\Pow_f X_0$ is the free semilattice 
on a set $X_0$ one has $TX = \Pow_f{X_0^*}$, the semilattice of 
all finite languages over $X_0$. Hence semilattice morphisms from $TX$ into 
the  two-chain $0<1$ correspond again to formal languages over $X_0$. This 
setting allows one to study \emph{disjunctive} varieties of languages in the sense 
of Pol\'ak~\cite{polak01}, see~\cite{Adamek2014c,Adamek2015,Adamek2015b}. 
Note that although the variety of idempotent semirings can also be represented 
by 
the free idempotent semiring monad  $T'X = \Pow_f X^*$ on $\Set$, 
\emph{functions} from $T'X = \Pow_f X^*$ to $\{0,1\}$ 
do not correpond to formal languages over $X$.  
\item \emph{Infinite words.} The monad
\[ TX = X^\infty = X^* + X^\omega \]
on $\D=\Set$ represents languages of finite and infinite words. The unit 
$\eta_X\colon X\to X^*$ is given by inclusion, and the multiplication
$\mu_X\colon 
(X^\infty)^\infty\to X^\infty$ is concatentation: $\mu_X (w_0w_1w_2\ldots)= 
w_0w_1w_2\ldots$ if all words $w_i$ are finite, and otherwise $\mu_X 
(w_0w_1w_2\ldots)= w_0w_1w_2\ldots w_j$ for the smallest $j$ with $w_j$ 
infinite. $\MT$-algebras are  \emph{$\infty$-monoids}, i.e.\ monoids with an 
additional 
$\omega$-ary multiplication and the expected mixed associative laws. Again, 
functions from $TX$ to $\{0,1\}$ correspond to languages (of finite and 
infinite words), and 
$\omega$-regular languages are precisely the languages recognised by finite 
$\infty$-monoids. This was observed by Boja\'nczyk~\cite{Bojanczyk2015}, who 
also derived an Eilenberg-type theorem for varieties of $\omega$-regular 
languages and pseudovarieties of
$\infty$-monoids along the lines of Wilke~\cite{wilke91}. As in (b) one can 
replace $\infty$-monoids in $\Set$ by 
``idempotent $\infty$-semirings'', viewed as algebras for a 
suitable monad on $\JSL$, and thus extend Pol\'ak's theorem \cite{polak01} 
from finite word languages to 
$\omega$-regular languages. We leave the details for future work.
\item In contrast to the previous examples, the category $\D^\MT$ is not always
monadic over $\Set$ resp. $\Pos$.  To see this, let $\D=\Set_{0,1}$ be the 
variety of sets with two constants, that is, the 
  category of all algebras over the signature with two constant symbols $0,1$. 
  The full subcategory $\Set_{0\neq 1}$, consisting of singletons and  sets 
  with distinct constants $0 \neq 1$, is reflective and hence monadic over 
  $\Set_{0,1}$. However, it is not monadic over $\Set$.
\end{enumerate}
\end{expl}

\section{Profinite Monads} \label{sec:profinite-monad}
In this section we introduce profinite monads, our main tool for 
the investigation of profinite equations and Reiterman's theorem for 
$\MT$-algebras in \Cref{sec:reiterman}.  

\begin{assumption}\label{asm:sec3}
  As in the previous section let $\D$ be a variety of algebras or ordered 
  algebras. Moreover, let $\MT=(T,\eta,\mu)$ be a monad on $\D$ such that $T$ 
  preserves surjective morphisms.
\end{assumption}
Recall that the \emph{right Kan extension} of a functor $F\colon \A \to \C$ 
along~$K\colon \A \to
\B$ is a functor~$R\colon \B \to \C$ with a universal natural
transformation $\epsilon\colon RK \to F$, i.e. for every functor
$G\colon \B \to \C$ and every natural transformation~$\gamma\colon GK \to F$ 
there
exists a unique natural transformation $\gamma^\dagger\colon G \to R$ with
$\gamma = \epsilon \o \gamma^\dagger K$. In case $F=K$, the functor $R$ carries
a natural monad structure: the unit is given by $\hateta=(\id_K)^\dagger: \Id\to
R$ and the multiplication by $\hatmu = (\epsilon\o R\epsilon)^\dagger: RR\to R$.
The monad $(R,\hateta,\hatmu)$ is called the \emph{codensity monad} of $K$, see
e.g., \cite{Linton1969}.

\begin{defn}\label{def:profmon}
  The \emph{profinite monad} of~$\MT$ is the codensity monad~$\hatT =
  (\hatt, \hateta, \hatmu)$ of the functor
  \[
    K=\hatJ U \colon \D_f^\MT \to \D_f 
    \to \hatD.
  \]
\end{defn}

\begin{rem}\label{rem:codensitymonad}
A related concept was recently studied by Boja\'nczyk~\cite{Bojanczyk2015} who
associates to every monad $\MT$ on $\Set$ a 
monad $\overline \MT$ on $\Set$ (rather than $\widehat{\Set}=\Stone$ as in our 
setting!).  
Specifically,
$\overline \MT$ is the monad induced by the composite right adjoint 
$\Stone^{\hatT} \to \Stone \xra{V} \Set$. Its construction also appears in the
work of Kennison and Gildenhuys~\cite{Kennison1971} who investigated codensity 
monads for $\Set$-valued functors and their connection with profinite algebras.
\end{rem}

\begin{rem}\label{rem:hattconst}
\begin{enumerate}[(a)]
\item One can compute $\hatt X$ for $X\in\hatD$ via the limit formula for 
right Kan extensions, see e.g.~\cite[Theorem X.3.1]{maclane}.  Letting $(X 
\downarrow \hatJ U)$ denote the comma category of all arrows $f\colon X\to A$ with 
$(A,\alpha)\in \D_f^\MT$, the object $\hatt X$ is the limit of the diagram
\[
  (X \downarrow \hatJ U) \to \hatD,\quad f\mapsto A.
\]
\item For $D\in\D$ a morphism $f\colon \hat D \to A$ with 
$(A,\alpha)\in\D^\MT_f$ corresponds to a $\MT$-homomorphism $h\colon (TD, 
\mu_D)\to(A, 
\alpha)$, since $(TD,\mu_D)$ is the free $\MT$-algebra on $D$.
  Hence to compute $\hatt \hat D$ one can replace $(\hat D \downarrow \hatJ 
  U)$ by
 the category of all $h\colon (TD, \mu_D)\to (A, \alpha)$ with $(A,\alpha)\in\D^\MT_f$. We denote the limit cone by
  \begin{equation}\label{eq:hatT-as-limit'}
    \br{h}\colon \hatt \hat{D} \to \hat{A}.
  \end{equation}
One can restrict the diagram defining $\hatt \hat D$ to surjective $\MT$-homomorphisms:
\end{enumerate}
\end{rem}

\begin{proposition}
  \label{prop:characterisation-hatT}
  For all $D\in \D$ the object $\hatt \hat{D}$ is the cofiltered limit of all 
  finite
      $\MT$-algebra quotients $e\colon (TD, \mu_D) \twoheadrightarrow (A, \alpha)$.
\end{proposition}

\begin{expl}[Profinite words] For the monad $TX=X^*$ on $\D=\Set$
the profinite monad $\hatT$ assigns to every finite set (= finite Stone space) 
$X$ 
the space $\hatt X = \widehat{X^*}$ of profinite words over $X$. This is the 
limit in $\Stone$ of all finite (discrete) quotient monoids of $X^*$. Similarly, for $TX = X^\infty$ the profinite monad $\hatT$ 
constructs the space $\hatt X$ of ``profinite $\infty$-words'' over $X$.
\end{expl}

\begin{lemma}\label{lem:hattcoflimsurj}
 \begin{enumerate}[(a)]
\item $\hatt$ preserves cofiltered limits and surjections.
\item Given a cofiltered limit cone $h_i\colon A\ra A_i$ ($i\in I$) in
    $\hatD^{\hatT}$, any $\hatT$-homomorphism $h\colon A\to B$ with finite 
    codomain 
    factors through 
    some $h_i$.
 \end{enumerate}
\end{lemma}

\begin{rem}\label{rem:homtheoremhatt}
\begin{enumerate}[(a)]
  \item Since $\hatt$ preserves surjections, the factorisation system of $\hatD$
    lifts to $\hatD^{\hatT}$, so we can speak about \emph{quotients} and
    \emph{subalgebras} of $\hatT$-algebras. Moreover, the homomorphism theorem
    holds for $\hatT$-algebras, cf.  \Cref{re:factorisation_DT}.
  \item \Cref{lem:hattcoflimsurj}(b) exhibits a crucial technical difference
    between our profinite monad $\hatT$ and Boja\'nczyk's $\overline{\MT}$, see
    \Cref{rem:codensitymonad}. For example, for the identity monad $\MT$
    on~$\Set$, the monad $\overline\MT$ is the ultrafilter monad whose algebras
    are compact Hausdorff spaces, and the factorisation property in the lemma
    fails.
\end{enumerate} 
\end{rem}

\begin{rem}\label{rem:alphaplus}
For each finite $\MT$-algebra $(A,\alpha)$ the morphism $\alpha$ is itself
  a $\MT$-homomorphism $\alpha\colon (TA,\mu_A)\epito (A,\alpha)$, and thus yields the limit projection  \[
      \br{\alpha} \colon \hatt \hat{A} \to \hat{A}
    \] of \eqref{eq:hatT-as-limit'}. The unit $\hat\eta_{\hat D}$ and multiplication $\hat \mu_{\hat D}$ of $\hatT$ are 
determined by the following commutative diagrams for all $\MT$-homomorphisms 
$h\colon (TD, \mu_D)\to (A, \alpha)$:
\begin{equation} \label{eq:hat-eta}
        \vcenter{
          \xymatrix{
            \hat{D} \ar[r]^{\hat{\eta}_{\hat{D}}} \ar[rd]_{\widehat{h \eta_D}}
            & \hatt \hat{D} \ar[d]^{\br{h}} &&  {\hatt \hatt \hat{D}} 
            \ar[r]^{\hat\mu_{\hat{D}}}
                                  \ar[d]_{\hatt \br{h}} & \hatt \hat{D}
                                  \ar[d]^{\br{h}} \\
            & \hat{A} && {\hatt \hat{A}} \ar[r]_{\br{\alpha}} & \hat{A}
          }
        }
\end{equation}
Hence $(\hat A, \alpha^+)$ is a $\hatT$-algebra: the unit and associative law for $\hatT$-algebras follow by putting $D =
A$ and~$h = \alpha$ in \eqref{eq:hat-eta}. Moreover, \eqref{eq:hat-eta}  states
precisely that $\br{h}\colon (\hatt \hat D, 
\hat\mu_{\hat D}) \to (A,\br{\alpha})$ is the unique $\hatT$-homomorphism 
extending the map $\widehat{h\eta_D}$ for every $h$ as above.
\end{rem}

\begin{proposition}\label{prop:finite-algebras}
 The maps $(A, \alpha) \mapsto (\hat{A}, \br{\alpha})$ and $h
  \mapsto \widehat{h}$ define an isomorphism between the categories of finite 
  $\MT$-algebras and finite $\hatT$-algebras:
   \[ \D_f^\MT \cong \hatD_f^{\hatT}. \]
\end{proposition}

\section{Reiterman's Theorem for $\MT$-Algebras}\label{sec:reiterman}

Reiterman's theorem~\cite{Reiterman1982,Banaschewski1983} states that, for any variety $\D$ of 
 algebras, a class of finite algebras in $\D$ is a 
\emph{pseudovariety},
i.e.\ closed under finite products, subobjects and quotients, iff it
is presented by \emph{profinite equations}. Later Pin 
and Weil~\cite{PinWeil1996}
proved the corresponding result for varieties $\D$ of ordered algebras:  
pseudovarieties are precisely the classes of finite algebras in $\D$
presented by \emph{profinite inequations}. In our categorical setting these two
theorems represent the case where $\MT$ is chosen to be the identity monad on 
$\D$. In \Cref{sec:pseudovar} we introduce pseudovarieties 
and profinite (in-)equations for arbitrary monads $\MT$ on $\D$, a 
straightforward extension of the original notions. In \Cref{sec:pseudovargen} we
present a further generalisation and prove the main result of this paper,
Reiterman's theorem for finite $\MT$-algebras.

\subsection{Pseudovarieties and profinite (in-)equations}\label{sec:pseudovar}

Let us start with extending the classical concept of a pseudovariety to $\MT$-algebras.

\begin{defn}\label{def:pseudovar}
A \emph{pseudovariety of $\MT$-algebras} is a class of finite $\MT$-algebras 
closed 
under finite products, subalgebras and quotients.
\end{defn}

\begin{notation}
Recall from \Cref{rem:homtheorem} the forgetful functor 
$\left|\mathord{-}\right|\colon \D\ra \Set$ and its left adjoint $X\mapsto 
\Phi_X$. For any finite $\MT$-algebra $(A,\alpha)$ to interpret variables from 
a finite set $X$ in 
$A$ means to give a morphism $h_0\colon \Phi_X\to A$ in $\D$, or equivalently a 
$\MT$-homomorphism 
$h\colon (T\Phi_X,\mu_{\Phi_X})\to (A,\alpha)$. The
corresponding $\hatT$-homomorphism is denoted
$\br{h}\colon \hatt\hat{\Phi_X} \to A$, 
see \Cref{rem:hattconst} and \ref{rem:alphaplus}.
\end{notation}

\begin{defn}
\begin{enumerate}
\item Let $\D$ be a variety of unordered algebras. By a \emph{profinite 
equation} over a finite set $X$ of variables is meant a pair 
$u,v\in\hatt\hat{\Phi_X}$, denoted $u=v$. A finite $\MT$-algebra $(A,\alpha)$ 
\emph{satisfies} $u=v$ provided that
\[ \br{h} (u) = \br{h}(v) \quad\text{for all $\MT$-homomorphisms } h\colon 
T\Phi_X\to A.\]
\item Let $\D$ be a variety of ordered algebras. A \emph{profinite 
inequation} over a finite set $X$ of variables is again a pair 
$u,v\in\hatt\hat{\Phi_X}$, denoted $u\leq v$. A finite $\MT$-algebra 
$(A,\alpha)$ 
\emph{satisfies} $u\leq v$ provided that
\[ \br{h}(u) \leq \br{h}(v) \quad\text{for all $\MT$-homomorphisms } h\colon 
T\Phi_X\to A.\]
\end{enumerate}
A class $E$ of profinite (in-)equations \emph{presents} the class of all 
finite $\MT$-algebras that satisfy all (in-)equations in $E$.
\end{defn}

\begin{lemma}\label{lem:eqtovar}
Every class of finite $\MT$-algebras presented by profinite 
(in-)equations forms a pseudovariety.
\end{lemma}
The proof is an easy verification. In the following subsection we show the 
converse of the lemma: every pseudovariety is presented by profinite equations.

\subsection{Reiterman's Theorem for $\MT$-algebras}\label{sec:pseudovargen}
The concept of profinite (in-)equation as introduced above only considers the
free finitely generated objects $\Phi_X$ of $\D$ as objects of variables. A
natural variation is to admit any finite object $X\in\D_f$ as an object of
variables. That is, we define a \emph{profinite equation over $X$} as a pair 
$u,v\in
\hatt\hat X$, and say that a finite $\MT$-algebra $(A,\alpha)$ \emph{satisfies}
$u=v$ if for every $\MT$-homomorphism $h\colon (TX,\mu_X)\to (A,\alpha)$
the $\hatT$-homomorphism $\br{h}\colon \hatt\hat X \ra A$ merges $u,v$;
analogously for inequations. A class of finite $\MT$-algebras presented by such
profinite equations is still closed under finite products and subalgebras, but
not necessarily under quotients. However, it is closed under \emph{$U$-split
  quotients} for the forgetful functor $U\colon \D^\MT_f \ra \D_f$, where a
surjective morphism $e$ in $\D^\MT_f$ is called \emph{$U$-split} if there is a 
morphism
$m$ in $\D_f$ with $Ue\o m = \id$.

More generally, we introduce below for a class $\X$ of objects in $\D$ the concept of
\emph{profinite (in-)equation} over $\X$: a pair of elements of $\hatt\hat X$
with $X\in\X$. This subsumes both of the above situations: by taking as $\X$
all free finitely generated objects of $\D$ we recover the concept of
\Cref{sec:pseudovar}. And the choice $\X=\D_f$ leads to a new variant of
Reiterman's theorem: a characterisation of classes of finite $\MT$-algebras
closed under finite products, subalgebras and $U$-split quotients. The latter
can be understood as a finite analogue of Barr's result~\cite{Barr2002}, which
states that classes of $\MT$-algebras closed under products, subalgebras and
$U$-split quotients are in bijective correspondence with quotient monads of
$\MT$.

\begin{notation}
For a class $\X$ of objects in $\D$ we denote by $\E_\X$ the class of all
surjective morphisms $e\colon A\epito B$ with finite codomain such that all objects
$X$ of $\X$ are projective w.r.t.~$e$. That is, every morphism $f\colon X\ra B$
factors through $e$.
\end{notation}

\begin{assumption}\label{asm:sec6}
We assume that a class $\X$ of objects in $\D$ is given that forms a
\emph{projective presentation} of $\D_f$, i.e.\ for every finite object 
$A\in\D_f$ there exists
an object $X\in\X$ and a quotient $e\colon X\epito A$ in $\E_\X$.
\end{assumption}

\begin{defn}
An \emph{$\X$-pseudovariety of $\MT$-algebras} is a class of finite 
$\MT$-algebras closed 
under finite products, subalgebras and $\E_\X$-quotients, i.e.\ quotients carried 
by a morphism in $\E_\X$.
\end{defn}

\begin{expl}\label{ex:quotients}
\begin{enumerate}[(a)]
  \item For the choice of \Cref{sec:pseudovar},
    \[
      \X =  \text{ free finitely generated objects of $\D$,}
    \]
    the class $\E_\X$ consists of all surjective morphisms with finite 
    codomain, see \Cref{rem:homtheorem}(c). Clearly \Cref{asm:sec6} is
    fulfilled since every finite object in a variety $\D$ is a quotient of a
    free finitely generated one. Thus an $\X$-pseudovariety is simply a 
    pseudovariety in the sense
    of \Cref{def:pseudovar}.
  \item If we choose
    \[
      \X = \D_f
    \]
    then $\E_\X$ consists precisely of the split surjections with finite
    codomain.  Indeed, clearly every split surjection lies in $\E_\X$.
    Conversely, given $e\colon A\epito B$ in $\E_\X$, apply the definition of
    $\E_\X$ to $X=B$ and $f=\id$. \Cref{asm:sec6} is fulfilled
    because every object in $\D_f$ is a split quotient of itself. A 
    $\D_f$-pseudovariety is a class of finite $\MT$-algebras closed under
    finite products, subalgebras and $U$-split quotients.
\end{enumerate}
\end{expl}

\begin{defn}
\begin{enumerate}
\item Let $\D$ be a variety of unordered algebras. A \emph{profinite equation over $\X$} is an expression 
of the form $u=v$ with $u,v\in \hatt \hatX$ and $X\in \X$. A finite 
$\MT$-algebra $(A,\alpha)$ \emph{satisfies} $u=v$  if
\[ \br{h}(u) = 
\br{h}(v) \quad\text{for all $\MT$-homomorphisms $h\colon TX\to A$.} \]
\item Let $\D$ be a variety of ordered algebras. A \emph{profinite inequation over $\X$} is an expression of 
the form $u\leq v$ 
with $u,v\in \hatt \hatX$ and $X\in \X$. A finite $\MT$-algebra 
$(A,\alpha)$ \emph{satisfies} $u\leq v$  if
\[ \br{h}(u) \leq
\br{h}(v) \quad\text{for all $\MT$-homomorphisms $h\colon TX\to A$.} \]
\end{enumerate}
A class $E$ of profinite (in-)equations over $\X$ \emph{presents} the class of all 
finite $\MT$-algebras that satisfy all (in-)equations in $E$.
\end{defn}

\begin{rem}\label{rem:phiv}
For any full subcategory $\V\seq \D^\MT_f$ closed under finite products and
subalgebras, the \emph{pro-$\V$ monad} of $\MT$ is the monad $\hatT_\V =
(\hatt_\V,\hatmu^\V, \hateta^\V)$ on $\hatD$ defined by replacing in
\Cref{def:profmon} the functor $U\colon \D^\MT_f \to \D_f$ by its restriction
$U_\V\colon \V\to \D_f$.  That is, $\hatt_\V$ is the right Kan extension of $\hatJ U_\V$ along itself. In analogy to \Cref{rem:hattconst}, one
can describe $\hatt_\V \hat{X}$ with $X\in\D$ as the cofiltered limit of the
diagram of all homomorphisms $h\colon (TX,\mu_X)\to (A,\alpha)$ with 
$(A,\alpha)\in \V$. The limit projections are denoted $\brv{h}\colon
\hatt_\V \hatX \ra A$.  The 
universal
property of $\hatt_\V$ as a right Kan extension yields a monad morphism
$\phi^\V\colon \hatT\ra \hatT_\V$; its component $\phi^\V_{\hatX}$ for
$X\in\DCat$ is the unique $\hatD$-morphism making the triangle below commute for all $h\colon (TX,\mu_X)\to (A,\alpha)$ 
with $(A,\alpha)\in \V$.
\begin{equation}\label{eq:phiv}
  \vcenter{
  \xymatrix{
  \hatt \hatX \ar@{->>}[r]^{\phi^\V_{\hatX}}\ar[d]_{\br{h}} & 
  \hatt_\V\hatX 
  \ar[dl]^{\brv{h}}\\
  A &
  }
  }
\end{equation}
\end{rem}

\begin{lemma}\label{lem:keylemma}
Let $\V$ be a class of finite $\MT$-algebras closed under finite products and 
subalgebras and $u,v\in \hatt\hatX$ with $X\in\D$.
\begin{enumerate}
\item Unordered case: $\phi^\V_\hatX 
(u) = \phi^\V_\hatX (v)$ iff every algebra in $\V$ satisfies 
$u=v$.
\item   Ordered case: $\phi^\V_\hatX 
(u) \leq \phi^\V_\hatX (v)$ iff every algebra in $\V$ satisfies $u\leq v$.
\end{enumerate}
\end{lemma}

\begin{theorem}[Reiterman's Theorem for $\MT$-algebras]\label{thm:reiterman}
A class of finite 
$\MT$-algebras is an $\X$-pseudovariety  iff it is presented
by profinite equations over $\X$ (unordered case) resp. profinite inequations 
over $\X$ (ordered case).
\end{theorem}

\begin{proof}
Consider first the unordered case. The ``if'' direction is a straightforward 
verification. For the ``only if'' direction let $\V$ be an
$\X$-pseudovariety.

\begin{enumerate}[(a)]
  \item In analogy to \Cref{prop:characterisation-hatT} one can restrict the
    cofiltered diagram defining $\hatt_\V \hatX$ to 
  surjective homomorphisms $h\colon TX\epito A$. Then the limit projections 
  $h_\V^+$ and the mediating map 
    $\phi^\V_{\hatX}$ in \eqref{eq:phiv} are also surjective,
    see~\cite[Corollary 1.1.6]{Ribes2010}. Moreover, since
    $\phi^\V$ is a monad morphism, the free $\hatT_\V$-algebra ($\hatt_\V \hatX,
    \hatmu^\V_{\hat{X}})$ on $\hatX$ can be turned into a $\hatT$-algebra
    $(\hatt_\V \hatX, \hatmu^\V_{\hatX}\o \phi^\V_{\hatt_\V\hatX})$, and
    $\phi^\V_{\hatX}\colon (\hatt \hatX,\hatmu_{\hatX}) \ra (\hatt_\V \hatX,
    \hatmu^\V_{\hatX} \o \phi^\V_{\hatt_\V \hatX})$ is a $\hatT$-homomorphism.
  \item Let $E$ the class of all profinite equations over $\X$ satisfied by all
    algebras in $\V$. We prove that $\V$ is presented by $E$, which only requires to show
    that every finite $\MT$-algebra $(A,\alpha)$ satisfying all equations in $E$
    lies in $\V$.

    By \Cref{asm:sec6} choose $X\in \X$ and a quotient $e_0\colon X\epito A$ 
    in $\E_\X$, and freely extend $e_0$ to a (necessarily
    surjective) $\MT$-homomorphism $e\colon TX\epito A$. We first show that the
     corresponding $\hatT$-homomorphism $\br{e}: \hatt \hatX \to \hat 
    A$ 
    factors through $\phi^\V_\hatX$. Indeed, whenever
    $\phi^\V_\hatX$ merges $u,v\in \hatt\hatX$ then the profinite equation $u=v$ lies in $E$ by
    \Cref{lem:keylemma}, so $\br{e}$ merges $u,v$ since
    $(A,\alpha)$ satisfies all equations in $E$. Since $\phi^\V_\hatX$ is
    surjective by (a), the homomorphism theorem (see \Cref{rem:homtheoremhatt})
    yields a $\hatT$-homomorphism $g\colon \hatt_\V\hatX\to A$ in $\hatD$ with 
    $g\o
    \phi^\V_\hatX = \br{e}$. 
	\item By \Cref{lem:hattcoflimsurj}(b) the $\hatT$-homomorphism $g$
    factors through the limit cone defining $\hatT_\V \hatX$: there is a
    $\MT$-homomorphism $h\colon TX\to B$ with $(B,\beta)\in\V$ and a
    $\hatT$-homomorphism $q\colon B\to  A$ with $q\o \brv{h} = g$.
    By \Cref{prop:finite-algebras} the morphism $q$ is also a 
    $\MT$-homomorphism, and is surjective
    because $g$ is. 
  \item To conclude the proof it suffices to verify that $q$ lies in $\E_\X$ (then $(B,\beta)\in\V$ implies $(A,\alpha)\in\V$ because $\V$ is closed under
    $\E_\X$-quotients). Indeed: every morphism $f\colon Y\to A$ with $Y\in\X$
    factors through $e_0$ because $e_0\in \E_\X$, i.e. \[ f = e_0\o
      k\quad\text{for some $k\colon Y\to X$ in $\D$.} \] Then the diagram below
    commutes (for the second triangle see \eqref{eq:hat-eta}) and
    shows that $\hat f$ factors through $\hat q = q$ in $\hatD$, so $f$ factors
    through $q$ in $\D$. We conclude that $q\in\E_\X$, as desired.
    \[
    \xymatrix{
    \hat Y \ar[dr]_{\hat f} \ar[r]^{\hat k} & \hatX 
    \ar@{->>}[d]_>>>>>>{\hat{e_0}} \ar[r]^{\hat\eta_{\hatX}}  & \hatt \hatX 
    \ar@{->>}[r]^{\phi_{\hatX}^\V} \ar@{->>}[dl]|{\br{e}}  & \hatt_\V\hatX 
    \ar@{->>}@/^1em/[dll]_g \ar[d]^{\brv{h}} \\
    & \hat A & & \hat B \ar@{->>}[ll]^{\hat q}
    }
    \]
\end{enumerate}
This proves the theorem for the unordered case. The proof for the ordered case
is analogous: replace profinite equations by inequations, and use the
homomorphism theorem for ordered algebras to construct the morphism~$g$.
\end{proof}

\section{Applications and Examples}

Let us consider some examples and applications. First note that the original Reiterman theorem and its
ordered version emerge from \Cref{thm:reiterman} by taking the identity monad
$\MT=\Id$ and $\X =$ free finitely generated objects of $\D$, see
\Cref{ex:quotients}(a). In this case we have $\hatT = \Id$, $\D^\MT = \D$,
$\hatD^{\hatT} = \hatD$, and a profinite equation $u=v$ (resp. a profinite
inequation $u\leq v$) is a pair $u,v\in \hat{\Phi_X}$ for a finite set $X$. We
conclude:

\begin{corollary}[Reiterman \cite{Reiterman1982}, Banaschewski \cite{Banaschewski1983}]
Let $\D$ be a variety of algebras. A class $\V\seq\D_f$ is a pseudovariety iff 
it is 
presented by profinite equations over finite sets of 
variables.
\end{corollary}

\begin{corollary}[Pin and Weil \cite{PinWeil1996}]
Let $\D$ be a variety of ordered algebras. A class $\V\seq\D_f$ is a 
pseudovariety iff it is presented 
by profinite inequations over finite sets of 
variables.
\end{corollary}

Recall from Isbell~\cite{isbell64} that a class $\V\seq
\D$ is closed under products and subalgebras iff it
is presented by implications 
\[
  \bigwedge_{i\in I} s_i=t_i ~\Ra~ s=t
\]
where $s_i$, $t_i$, $s$, $t$ are terms and $I$ is a set.  Choosing $\MT$ to be
the identity monad and $\X=\D_f$ gives us the counterpart for finite algebras:
by \Cref{ex:quotients}(b) a $\D_f$-pseudovariety is precisely a class
$\V\seq\D_f$ closed under finite products and subalgebras, since the closure
under split quotients is implied by closure under subalgebras.  Such a class
could be called ``quasi-pseudovariety'', but to avoid this clumsy terminology we
prefer ``quasivariety of finite algebras''.

\begin{defn}
A \emph{quasivariety of finite algebras of $\D$} is a class $\V\seq\D_f$ closed under finite products and subalgebras.
\end{defn}

In analogy to Isbell's result we show that quasivarieties of finite algebras 
are precisely the classes of finite algebras 
of $\D$ presented by \emph{profinite implications}.

\begin{defn}
Let $X$ be a finite set of variables.
\begin{enumerate}
\item Unordered case: a \emph{profinite implication} over $X$ is an expression 
\begin{equation}\label{eq:profimpunord}
\bigwedge_{i\in I} u_i=v_i ~\Ra~ u=v
\end{equation}
 where $I$ is a set and 
$u_i, v_i, u, v\in  \hat{\Phi_X}$. An object $A\in\D_f$ 
\emph{satisfies} \eqref{eq:profimpunord} if for every $h\colon \Phi_X\ra A$ 
with 
$\hat{h}(u_i)=\hat{h}(v_i)$ for all $i\in I$ one has 
$\hat{h}(u)=\hat{h}(v)$.
\item Ordered case: a \emph{profinite implication} over $X$ is an expression
\begin{equation}\label{eq:profimpord}
\bigwedge_{i\in I} u_i\leq v_i ~\Ra~ u\leq v
\end{equation}
 where $I$ is a set and 
$u_i, v_i, u, v\in  \hat{\Phi_X}$. An object $A\in\D_f$ 
\emph{satisfies} \eqref{eq:profimpord} if for every $h\colon \Phi_X\ra A$ with 
$\hat{h}(u_i)=\hat{h}(v_i)$ for all $i\in I$ one has 
$\hat{h}(u)\leq \hat{h}(v)$.
\end{enumerate}
A class $P$ of profinite implications \emph{presents} the class of all finite algebras in $\D$ satisfying all implications in $P$.
\end{defn}

\begin{theorem}\label{cor:reitermanquasi}
For any class $\V\seq \D_f$ the following statements are equivalent:
\begin{enumerate}
\item $\V$ is a quasivariety of finite algebras.
\item $\V$ is presented by profinite (in-)equations over $\D_f$.
\item $\V$ is presented by profinite implications.
\end{enumerate}
\end{theorem}

\renewcommand*{\proofname}{Proof sketch}
\begin{proof}
3$\Ra$1 requires a routine verification, and 1$\Ra$2 is \Cref{thm:reiterman}. For 2$\Ra$3 
assume w.l.o.g. that $\V$ is presented by a 
single profinite equation $u=v$ with $u,v$ elements of some $X\in\D_f$. 
Express $X$ as a quotient $q\colon \Phi_Y\epito X$  for some finite set $Y$. 
Let 
$\{\,(u_i,v_i) : i \in I\,\}$ be the kernel of $\hat q\colon  \hat{\Phi_Y} 
\epito 
X$ (consisting of all pairs $(u_i,v_i)\in \hat{\Phi_Y}\times\hat{\Phi_Y}$ with 
$\hat q (u_i) = \hat q (v_i)$), and choose $u',v'\in \hat{\Phi_Y}$ with $\hat 
q(u') = u$ and  $\hat 
q(v') = v$.  Then a finite object $A\in\D_f$ satisfies the 
profinite equation $u=v$ iff it satisfies the profinite implication 
\begin{equation}\label{eq:imp}
\bigwedge_{i\in I} u_i=v_i ~\Ra~ u'=v',
\end{equation}
which proves that $\V$ is presented \eqref{eq:imp}. Analogously for the ordered case.
\end{proof}
\renewcommand*{\proofname}{Proof}

\begin{expl}
\begin{enumerate}
\item Let $\V\seq \mathbf{Mon}_f$ be the quasivariety of all finite monoids
  whose only invertible element is the unit. It is presented by the profinite
  implication $x^\omega = 1 \Ra x=1$ over the set of variables $X=\{x\}$. Here
  the profinite word $x^\omega\in \widehat{X^*}$ is interpreted, for every 
  finite monoid $M$ with $x$ interpreted as $m\in M$, as the unique idempotent 
  power of $m$. Indeed, if $M$ has no nontrivial invertible elements, it 
  satisfies the implication: given $m\neq 1$ and $m^k$ idempotent, then 
  $m^k\neq 1$ (otherwise $m$ has the inverse $m^{k-1}$). Conversely, if $M$ 
  satisfies the implication and $m$ is invertible, then so is its idempotent 
  power 
  $m^k$. Hence $m^k\o m^k = m^k$ implies $m^k=1$, so
  $m=1$.
\item Let $\Pos$ be the variety of posets (i.e.\ the variety of all ordered 
algebras over the empty signature). The quasivariety $\V\seq \Pos_f$ of 
finite 
discrete 
posets is presented by the profinite implication $v\leq u \Ra u\leq 
v$ over the set $X=\{u,v\}$.
\end{enumerate}
\end{expl}

\begin{rem}\label{rem:many-sorted}
As indicated before all concepts in this paper also apply to a setting where 
$\D$ is a many-sorted variety of algebras or ordered algebras.  In this case 
an algebra is \emph{finite} if the disjoint union of the underlying sets of 
all sorts is a 
finite set. By a 
\emph{profinite equation} over $X\in \D$ is a meant pair of elements $u,v$ in 
some 
sort $s$ 
of 
$\hatt \hatX$, and it is \emph{satisfied} by a finite $\MT$-algebra $A$ if for 
every 
$\MT$-homomorphism $h\colon TX\to A$ the $s$-component of $h^+\colon \hatt\hatX\to A$ 
merges $u,v$. Similarly for 
profinite inequations and profinite implications. 
\end{rem}

\begin{expl}\label{ex:jones}
Consider the variety $\D$ of directed graphs, i.e. algebras for the two-sorted 
signature consisting of a sort $\mathsf{Ob}$ (objects), a sort $\mathsf{Mor}$ 
(morphisms) and two 
unary operations $s,t\colon \mathsf{Mor}\to \mathsf{Ob}$ specifying the 
source and target of a 
morphism. 
Then $\mathbf{Cat}$, the category of small categories and functors, is 
isomorphic to $\D^\MT$ 
for the 
monad $\MT$ constructing the free category on a 
graph. Choosing $\X = $ free finitely generated graphs, \Cref{thm:reiterman} shows that every
pseudovariety of categories, i.e.\ every class of finite 
categories closed under finite products, subcategories (represented by 
injective functors) and quotient 
categories (represented by surjective functors), can be specified by profinite 
equations over a two-sorted set of variables. This result was essentially 
proved by Jones~\cite{Jones1996}.
The difference is that he restricts to quotients represented 
by  
surjective functors which are bijective on objects, and replaces subcategories 
by faithful functors. Moreover, 
profinite equations are restricted to the sort of morphisms.
\end{expl}

\section{Conclusions and Future Work}
Motivated by recent developments in algebraic language theory, we generalised 
Reiterman's theorem to finite algebras for an arbitrary monad $\MT$ on a base 
category $\D$. Here $\D$ is a variety of (possibly ordered, many-sorted) 
algebras. The core concept of our paper is the profinite monad $\hatT$ of 
$\MT$, which makes it possible to introduce profinite (in-)equations at the 
level of monads and prove that they precisely present pseudovarieties of 
$\MT$-algebras.

Referring to the diagram in the Introduction, our Reiterman theorem is presented
in a setting that unifies the two categorical approaches to algebraic language
theory of Boja\'nzcyk \cite{Bojanczyk2015} and in our work
\cite{Adamek2014c,Adamek2015,Adamek2015b, Chen2015}. The next step is to also
derive an Eilenberg theorem in this setting. For each monad $\MT$ on a category
of sorted sets, Boja\'nczyk~\cite{Bojanczyk2015} proved an Eilenberg-type
characterisation of pseudovarieties of $\MT$-algebras: they correspond to
\emph{varieties of $\MT$-recognisable languages}. Here by a ``language'' is
meant a function from $TX$ to $\{0,1\}$ for some alphabet $X$, and a variety of
languages is a class of such languages closed under boolean operations,
homomorphic preimages and a suitably generalised notion of derivatives. On the
other hand, as indicated in \Cref{ex:monads}, one needs to consider
monoids on algebraic categories beyond $\Set$ in order to study varieties of
languages with relaxed closure properties, e.g. dropping closure under
complement or intersection. The aim is thus to prove an Eilenberg theorem
parametric in a monad $\MT$ on an algebraic category $\D$. Observing that e.g.
for $\D=\Set$ the monad $\hatT$ on $\Stone$ dualises to a \emph{comonad} on the
category of boolean algebras, we expect this can be achieved in a duality-based
setting along the lines of Gehrke, Grigorieff and Pin~\cite{Gehrke2008} and our
work~\cite{Adamek2014c,Adamek2015}.

 Throughout this paper we presented the case of ordered and unordered algebras as 
separated but analogous developments. Pin and Weil~\cite{PinWeil1996} gave a 
uniform treatment of ordered and unordered algebras by generalising  
 Reiterman's theorem from finite algebras to finite first-order 
structures. A similar  approach should also work in our 
categorical framework: replace $\D$ by a variety of relational algebras over a 
quasivariety $\Q$ of relational first-order structures, with $\Q=\Set$ and 
$\Q=\Pos$ covering the case of algebras and ordered algebras.

Finally, observe that categories of the form $\D^\MT$, where $\D$ is a many-sorted variety of algebras and $\MT$ is an accessible monad, correspond precisely to locally 
presentable categories. This opens the door towards an abstract treatment, and 
further generalisation, of Reiterman's theorem in purely categorical 
terms.

%%% Local Variables:
%%% mode: latex
%%% TeX-master: "pseudovariety-conf"
%%% End:

\bibliographystyle{splncs03}
\bibliography{reference}

\begin{thebibliography}{10}
\providecommand{\url}[1]{\texttt{#1}}
\providecommand{\urlprefix}{URL }

\bibitem{arv11}
Ad\'{a}mek, J., Rosick\'y, J., Vitale, E.: Algebraic Theories. Cambridge
  University Press (2011)

\bibitem{Adamek2014c}
Ad\'{a}mek, J., Milius, S., Myers, R.S.R., Urbat, H.: {Generalized Eilenberg
  theorem I : Local varieties of languages}. In: Muscholl, A. (ed.)
  Proc.~FoSSaCS'14. LNCS, vol. 8412, pp. 366--380. Springer Berlin Heidelberg
  (2014)

\bibitem{Adamek2015b}
Ad\'{a}mek, J., Milius, S., Urbat, H.: Syntactic monoids in a category. In:
  Moss, L.S., Sobocinski, P. (eds.) Proc.~CALCO'15 (2015)

\bibitem{Adamek2015}
Ad\'{a}mek, J., Myers, R.S.R., Urbat, H., Milius, S.: {Varieties of languages
  in a category}. In: Proc.~LICS'15. IEEE (2015)

\bibitem{Adamek2004a}
Ad\'{a}mek, J., Porst, H.E.: {On tree coalgebras and coalgebra presentations}.
  Theor. Comput. Sci.  311(1-3),  257--283 (2004)

\bibitem{Adamek1994}
Ad\'{a}mek, J., Rosick\'{y}, J.: {Locally Presentable and Accessible
  Categories}. Cambridge University Press (1994)

\bibitem{almeida90}
Almeida, J.: On pseudovarieties, varieties of languages, filters of
  congruences, pseudoidentities and related topics. Algebra universalis  27(3),
   333--350 (1990)

\bibitem{Banaschewski1983}
Banaschewski, B.: {The Birkhoff Theorem for varieties of finite algebras}.
  Algebra universalis  17(1),  360--368 (1983)

\bibitem{Barr2002}
Barr, M.: {HSP subcategories of Eilenberg-Moore algebras}. Theory Appl. Categ.
  10(18),  461--468 (2002)

\bibitem{Bloom1976}
Bloom, S.L.: {Varieties of ordered algebras}. J. Comput. Syst. Sci.  13(2),
  200--212 (1976)

\bibitem{Bojanczyk2015}
Bojańczyk, M.: Recognisable languages over monads. In: Potapov, I. (ed.)
  Proc.~DLT, LNCS, vol. 9168, pp. 1--13. Springer (2015), full version:
  \url{http://arxiv.org/abs/1502.04898}

\bibitem{Chen2015}
Chen, L.T., Urbat, H.: {A fibrational approach to automata theory}. In: Moss,
  L.S., Sobocinski, P. (eds.) Proc.~CALCO'15 (2015)

\bibitem{Eilenberg1976}
Eilenberg, S.: {Automata, Languages, and Machines}, vol.~2. Academic Press, New
  York (1976)

\bibitem{Gehrke2008}
Gehrke, M., Grigorieff, S., Pin, J.E.: {Duality and equational theory of
  regular languages}. In: Aceto, L., Damg\aa~rd, I., Goldberg, L.A.,
  Halld\'{o}rsson, M.M., Ing\'{o}lfsd\'{o}ttir, A., Walukiewicz, I. (eds.)
  Proc.~ICALP'08, Part II, LNCS, vol. 5126, pp. 246--257. Springer Berlin
  Heidelberg (2008)

\bibitem{isbell64}
Isbell, J.R.: Subobjects, adequacy, completeness and categories of algebras.
  Instytut Matematyczny Polskiej Akademi Nauk (1964)

\bibitem{Johnstone1982}
Johnstone, P.T.: {Stone spaces}. Cambridge University Press (1982)

\bibitem{Jones1996}
Jones, P.R.: {Profinite categories, implicit operations and pseudovarieties of
  categories}. J. Pure Appl. Algebr.  109(1),  61--95 (1996)

\bibitem{Kennison1971}
Kennison, J.F., Gildenhuys, D.: {Equational completion, model induced triples
  and pro-objects}. J. Pure Appl. Algebr.  1(4),  317--346 (1971)

\bibitem{Linton1969}
Linton, F.E.J.: {An outline of functorial semantics}. In: Eckmann, B. (ed.)
  Semin. Triples Categ. Homol. Theory, LNM, vol.~80, pp. 7--52. Springer Berlin
  Heidelberg (1969)

\bibitem{maclane}
{Mac Lane}, S.: Categories for the working mathematician. Springer, 2 edn.
  (1998)

\bibitem{Pin1995}
Pin, J.E.: {A variety theorem without complementation}. Russ. Math. (Izvestija
  vuzov.Matematika)  39,  80--90 (1995)

\bibitem{PinWeil1996}
Pin, J.E., Weil, P.: { A Reiterman theorem for pseudovarieties of finite
  first-order structures}. Algebra Universalis  35,  577--595 (1996)

\bibitem{polak01}
Pol\'ak, L.: Syntactic semiring of a language. In: Sgall, J., Pultr, A.,
  Kolman, P. (eds.) Proc.~MFCS. LNCS, vol. 2136, pp. 611--620. Springer (2001)

\bibitem{Priestley1972}
Priestley, H.A.: {Ordered topological spaces and the representation of
  distributive lattices}. Proc. London Math. Soc.  3(3),  507 (1972)

\bibitem{Reiterman1982}
Reiterman, J.: {The Birkhoff theorem for finite algebras}. Algebra Universalis
  14(1),  1--10 (1982)

\bibitem{reutenauer80}
Reutenauer, C.: S\'eries formelles et alg\`ebres syntactiques. J.~Algebra  66,
  448--483 (1980)

\bibitem{Ribes2010}
Ribes, L., Zalesskii, P.: {Profinite Groups}. A Series of Modern Surveys in
  Mathematics, Springer Berlin Heidelberg (2010)

\bibitem{sch65}
Sch\"utzenberger, M.P.: On finite monoids having only trivial subgroups.
  Inform. and Control  8,  190--194 (1965)

\bibitem{Street1972}
Street, R.: {The formal theory of monads}. J. Pure Appl. Algebr.  2(2),
  149--168 (1972)

\bibitem{wilke91}
Wilke, T.: An {Eilenberg} theorem for infinity-languages. In: Proc.~ICALP'91.
  pp. 588--599 (1991)

\end{thebibliography}

% appendix for ArXiv preprint version. 
\appendix
\section{Topological toolkit}
The proofs of the following lemmas can be found in Chapter~1 of~\cite{Ribes2010}.
\begin{lemma}
  \label{lem:surjections-between-cofiltered-diagrams}
  %\label{coro:cone-of-surjections-is-surjective}
  Let $\tau\colon D_1 \to D_2$ be a natural transformation between cofiltered
  diagrams (of the same index) in the category of compact Hausdorff
  spaces. If each $\tau_i\colon D_1i \twoheadrightarrow D_2i$ is surjective, 
  so is the mediating map
  \[
    \Lim \tau \colon \Lim D_1 \to \Lim D_2.
  \]
  In particular, if $\tau_i\colon X \twoheadrightarrow Di$ is a cone of
  surjections, then the mediating map $h = \lim\tau$ is surjective.
\end{lemma}
\begin{lemma}
  \label{lem:dense-lemma}
  Let $(X \xrightarrow{\rho_i} X_i)_i$ be a limit of a cofiltered diagram
  in the category of topological spaces, and $(Y \xrightarrow{f_i} X_i)_i$ a 
  cone of that diagram consisting
  of surjective
  continuous functions. Then, the mediating map
  \[
    f\colon Y \to X 
  \]
  is dense, i.e.\ the image~$f[Y]$ is a dense subset of~$X$.
\end{lemma}
\begin{lemma}
  \label{prop:projection-is-surjective}
  Let $D$ be a cofiltered diagram in the category of compact Hausdorff spaces. 
  If all $D(i \xrightarrow{f} j)$
  are surjective, so is each projection $\rho_i\colon \Lim D \to Di$.
\end{lemma}

\section{Detailed proofs}

This appendix contains all proofs and additional details we omitted due to space restrictions.\\

\begin{defn}
An object $A$ in a category $\A$ is called \emph{finitely copresentable} if 
the hom-functor $\A(\mathord{-},A): \A^{\op}\to\Set$ preserves filtered 
colimits. Equivalently, given a cofiltered limit cone $\pi_i: B\to B_i$ ($i\in 
I$) in $\A$, any morphism $f: B\to A$ factors essentially uniquely through 
some $\pi_i$. The category $\A$ is \emph{locally finitely copresentable} if 
(i) $\A$ is complete, (ii) $\A$ has only a set of finitely copresentable 
objects up to isomorphism, and (iii) every object of $A$ is a cofiltered limit 
of finitely copresentable objects.
\end{defn}

\begin{rem}\label{rem:lfcp}
For any small category $\A_0$ with finite limits the pro-completion $\A=\Pro{\A_0}$, i.e. the free completion of $\A_0$ under cofiltered limits, is locally finitely copresentable, and the full subcategory of finitely copresentable objects is equivalent to $\A_0$. See e.g. \cite[Theorem 6.23]{arv11}.
\end{rem} 

\noindent\emph{Details for Remark \ref{re:hatD-is-procompletion}}. By \cite[Remark VI.2.4]{Johnstone1982} the category $\hatD$ is the pro-completion of $\D_f$. (The argument given there is for varieties of single-sorted algebras, but also applies to the ordered and many-sorted case.) Hence $\hatD$ is locally finitely copresentable and the objects of $\D_f$ are precisely the finitely copresentable objects of $\hatD$, see  \Cref{rem:lfcp}. This implies (ii). For (i) see \cite[Proposition 1.22]{Adamek1994}.

\begin{lemma}\label{lem:hatdlimits}
$\hatD$ is closed under cofiltered limits in the category of topological 
$\D$-algebras. In particular, the forgetful functor $V: \hatD\to \D$ preserves 
cofiltered limits.
\end{lemma}

\begin{proof}
Let $\Top(\D)$ denote the category of topological $\D$-algebras. Since $\hat 
J: \D_f \monoto \hatD$ is the pro-completion of $\D_f$ the inclusion functor 
$I: \D_f\monoto \Top(\D)$ has an essentially unique extension $I': \hatD\to 
\Top(\D)$ preserving cofiltered limits. The functor $I'$ maps an object $D\in 
\hatD$ to the cofiltered limit $I'D$ of the diagram of all morphisms $f: D\to 
D'$ with finite codomain. But since $D$ is profinite it follows $I'D \cong D$, 
that is, $I'$ is just the inclusion.

That $V$ preserves cofiltered limits now follows from the fact that limits in $\Top(\D)$ are formed on the level of $\D$.
\end{proof}

\begin{rem}\label{rem:arrowcatclfp}
 Since $\hatD$ is locally finitely copresentable so is the arrow category 
 $\hatD^\to$. Its finitely copresentable objects are precisely the morphisms 
 in $\D_f$, i.e. morphisms with finite domain and codomain, see 
 \cite[Corollary 1.54 and Example 1.55]{Adamek1994}. Hence every morphism 
 $f\colon A \to B$ in~$\hatD$ is a cofiltered limit (taken in $\hatD^\to$) of 
 morphisms in $\D_f$.
\end{rem}

\begin{proof}[Lemma \ref{lem:hatD-factorisation}]
 Let $f\colon A \to B$ in~$\hatD$. Express $f$ as a cofiltered limit of 
 morphisms $f_i$ in $\D_f$ with the
  limit cone $(a_i, b_i)\colon f \to f_i$, see \Cref{rem:arrowcatclfp}. 
  Factorise $f_i=m_i\o e_i$ into a quotient followed by a subalgebra in $\D$.
  \[
    \vcenter{
      \xymatrix{
        A \ar[rr]^f  \ar[d]_{a_i} &  & B \ar[d]^{b_i} \\
        A_i \ar@{->>}[r]_{e_i} \ar `d[r]`[rr]_{f_i}[rr] & C_i
        \ar@{ >->}[r]_{m_i} & B_i.
      }
    }
  \]
  Due to the diagonal fill-in property, the finite objects $C_i$ form  a
  cofiltered diagram with $(e_i)$ and $(m_i)$ natural transformations. Let  $\left(c_i\colon C
    \to C_i\right)$ be its limit in $\hatD$, and, $e \defeq \Lim e_i$, and $m \defeq \Lim m_i$ in $\hatD^\to$,
  as shown in the diagram below:
  \[
    \vcenter{
      \xymatrix{
        A \ar`u[r]`r[rr]^{f}[rr] \ar[r]^{e} \ar[d]_{a_i} & C \ar[d]^{c_i}
        \ar[r]^{m} & B
        \ar[d]^{b_i} \\
        A_i \ar`d[r]`r[rr]_{f_i}[rr] \ar@{->>}[r]^{e_i} & C_i
        \ar@{ >->}[r]^{m_i} & B_i.
      }
    }
  \]
  The morphism $e$ is surjective by 
  \Cref{lem:surjections-between-cofiltered-diagrams}. To show that $m$ is 
  injective observe that $m(c) = m(c')$ implies
  \[ m_i\o c_i(c) = b_i\o m(c) = b_i\o m(c') = m_i\o c_i(c') \]
   thus $c=c'$ since $m_i$ is monic and the morphisms $c_i$ are jointly monic. 
   In the ordered case, to show that $m$ is order-reflecting use instead that 
   $m_i$ is order-reflecting and the morphisms $c_i$ are jointly 
   order-reflecting.
\end{proof}

\begin{proof}[\Cref{prop:profcomp}]
 Since $V$ preserves cofiltered limits by Lemma \ref{lem:hatdlimits}, there is 
 a 
 unique mediating morphism $\eta_D: D \to V\hat D$ in $\D$ with $V\hat h \o 
 \eta_D = h$ for all $h: D\to D'$ with $D'\in\D_f$. We need to show that for 
 any $g: D\to VE$ with $E\in\hatD$ there is a unique $g^@: \hat D\to E$ with 
 $Vg^@\o \eta_D = g$. 
 
 Express $E$ as a cofiltered limit $p_i\colon E\ra E_i$, $i\in 
    I$, with $E_i\in\D_f$. Then the morphisms $\widehat{Vp_i\o g}: \hat D \to 
    E_i$ form a cone, so there is a unique $g^@\colon \hat{D}\ra E$ 
    with $p_i\o g^@ = \widehat{Vp_i\o g}$ for all $i$. It follows that
   \[ Vp_i \o Vg^@ \o \eta_D = V(\widehat{Vp_i\o g})\o \eta_D = Vp_i\o g. \] 
    Since $V$ preserves cofiltered limits the morphisms $Vp_i$ are jointly 
    monic, so $Vg^@\o \eta_D = g$. Moreover, since one can restrict the limit 
    cone defining $\hat D$ to surjective morphisms $h: D\epito D'$ (they form 
    an 
    initial subdiagram of $D\downarrow \D_f$), the morphism 
    $\eta_D$ is dense by \Cref{lem:dense-lemma} (where $D$ is viewed as a 
    discrete 
    topological space). This implies that the morphism $g^@$ is uniquely 
    determined by the equation $Vg^@ \o \eta_D = g$, as desired.
\[
\xymatrix{
D \ar[dr]_g\ar[r]^{\eta_D} \ar@/^2em/[rr]^f & V\hat{D} \ar[d]_{Vg^@} 
\ar[r]^<<<<<{V\hat h} \ar[dr]^>>>>{V(\widehat{Vp_i\o g})} & VD' = D' \\
& VE \ar[r]_<<<<<{Vp_i} & VE_i = E_i
}
\]
We conclude that $V$ has a left adjoint whose action on objects is given by 
$D\mapsto \hat D$, and on morphisms $h: D\to E$ by $h\mapsto (\eta_E\o h)^@$. 
It remains to show that for $E$ finite we have $(\eta_E\o g)^@ = \hat g$. Indeed, in this case we have $\hat E = E$ and $\eta_E = \id$, and the limit cone $(p_i)$  can be chosen trivial (that is, $I=\{1\}$, $E_1=E$ and $p_1 = \id$). Then
\[ (g\eta_E)^@ = g^@ = p_1\o g^@ = \widehat{Vp_1\o g} = \hat g.\]
\end{proof}

\noindent\emph{Details for \Cref{re:factorisation_DT}.} We prove the claim that the factorisation system of $\D$ of lifts to $\D^\MT$. Letting 
  $h=m\o e$ be the canonical factorisation of $h$ in $\D$, diagonal fill-in 
  gives a unique $\gamma: TC\to C$ making the diagram below commute. One 
  easily verifies that $(C,\gamma)$ is a $\MT$-algebra, so $e: 
  (A,\alpha)\epito 
  (C,\gamma)$ and $m: (C,\gamma)\monoto (B,\beta)$ are $\MT$-homomorphisms.
  \[
    \xymatrix{
      TA \ar@{->>}[r]^{Te} \ar@<-.1ex> `u[]`[rr]^{Th}[rr]
      \ar[d]_{\alpha} & \ar[r]^{Tm} TC
      \ar@{-->}[d]^{\gamma} & TB \ar[d]^{\beta} \\
      A \ar@{->>}[r]_{e} \ar@<-.1ex> `d[]`[rr]_{h}[rr] & C \ar@{ >->}[r]_m & 
      B
    }
  \] 

\begin{notation}\label{not:limitcone}
Recall from Remark \ref{rem:hattconst} that for $X\in\hatD$ the object $\hatt 
X$ 
is the limit of the diagram
\[ (X \downarrow \hatJ U) \to \hatD,\quad (f: X\to \hatJ U(A,\alpha)) \mapsto 
A. \]
We denote the limit projections by $f^* : \hatt X \to A$.
\end{notation}

\noindent\emph{Details for Remark \ref{rem:hattconst}.} We show that the 
diagrams \eqref{eq:hat-eta} commute for all $h: (TD,\mu_D) \to (A,\alpha)$ 
with $(A,\alpha)\in \D^\MT_f$.  As for the left hand diagram, recall the 
definition of $\hateta = (\id_{\hatJ 
  U})^\dagger$
  in the limit formula for right Kan extensions: for each $f\colon X \to 
  \hat{J} U(A,
  \alpha)$, the collection of $f$'s forms a compatible cone for
  the limit defining $\hatt X$, and $\hateta_X$ is the unique mediating 
  morphism satisfying the diagram
  \[
    \xymatrix{
      X \ar[dr]_{f} \ar[r]^{\hateta_X} & \hatt X \ar[d]^{f^*} \\
       & \hat{A}.
    }
  \]
  For $X=\hat D$ and $f=\widehat{h\eta_D}$ we have $f^* = h^+$, which  proves 
  the claim. Similarly for the right hand diagram.

\begin{proof}[\Cref{prop:characterisation-hatT}]
  By \Cref{re:factorisation_DT}, 
  we can factorise every morphism $h\colon (TD, \mu_D) \to (A, \alpha)$ as a
  quotient followed by a subalgebra:
  \[
    \xymatrix{
      TTD \ar@{->>}[r]^{Te} \ar[d]_{\mu_D} & \ar[r]^{Tm} TA_0
      \ar@{-->}[d]^{\alpha_0} & TA \ar[d]^{\alpha} \\
      TD \ar@{->>}[r]_{e} & A_0 \ar@{ >->}[r]_m & A.
    }
  \]
  This shows the subdiagram of all quotients of $(TD, \mu_D)$ is initial in the
  diagram defining~$\hatt \hat D$. Hence their limits are the same.
\end{proof}

%\begin{lemma} \label{lem:pi-ae-surjective}
%  For each $D\in\D$, the free $\hatT$-algebra $(\hatt \hat D, 
%  \hat\mu_{\hat D})$ 
%  is a
%  cofiltered limit in~$\hatD^{\hatT}$ with the limit cone $\br{e}$
%  formed by surjective $\MT$-homomorphisms $e: (TD,\mu_D) \epito (A,\alpha)$ 
%  with $(A,\alpha)\in \D^\MT_f$.
%\end{lemma}

\begin{proof}
By \Cref{prop:characterisation-hatT}  we have the limit cone 
$e^+: \hatt \hat D \to A$ in $\hatD$. That all $\br{e}$'s are surjective 
follows by
  \Cref{prop:projection-is-surjective}.
  Further, the morphisms $\br{e}$ are $\hatT$-algebra homomorphisms and hence
  $e^+:(\hatt
  X, \hat\mu_X)\to (A, \alpha^+)$ is a cofiltered limit cone in 
  $\hatD^{\hatT}$ 
  by~\eqref{eq:hat-eta}, 
  ,\Cref{rem:alphaplus} and the fact that limits of $\hatT$-algebras are 
  formed on the level of $\hatD$. 
\end{proof}

\begin{proof}[\Cref{lem:hattcoflimsurj}]
(a) $\hatt$ preserves cofiltered limits: for the purpose of this proof it is 
convient to express the limit defining $\hatt X$ for $X\in\hatD$, see 
\Cref{not:limitcone}, by the end 
formula
\[ \hatt X = \int_{(A, \alpha)\in \D_f^\MT} \widehat{\D}(X, \hat{A}) 
\pitchfork \hat{A}.\]
  Note that the power functor $- \pitchfork B\colon \Set^\op \to \hatD$
  preserves limits, and for $A$ finite the hom-functor $\hatD(-, 
  \hat{A})\colon 
  \hatD^\op
  \to \Set$ preserves filtered colimits. Hence the composite $\hatD(-, \hat{A})
  \pitchfork B\colon \hatD \to \hatD$ preserves cofiltered limits. Expressing
  $\hatt$ as an end, it follows that $\hatt$ preserves
  cofiltered limits:
  \begin{align*}
    \hatt (\Lim X_i) 
     = & \int_{(A, \alpha)} \widehat{\D}(\Lim X_i, \hat{A})
    \pitchfork \hat{A}  \\
     = &\int_{(A, \alpha)} \Lim_i \left(\widehat{\D}(X_i, \hat{A}) \pitchfork
       \hat{A} \right)  \\
     = &\Lim_i \left(\int_{(A, \alpha)} \widehat{\D}(X_i, \hat{A}) \pitchfork
       \hat{A} \right) \\
     = &\Lim_i (\hatt X_i).
  \end{align*}

\noindent $\hatt$ preserves surjective morphisms: 
  let $f\colon X \to Y$
  be a surjective morphism of~$\hatD$. We need to show that $\hatt f$ is 
  surjective.
  \begin{enumerate}[(1)]
    \item Suppose first that $X$ and $Y$ are finite.  Then, $f = 
    \widehat{f_0}$ for 
    some
      $f_0\colon X_0 \to Y_0$ in~$\D$.  Since $T$ preserves 
      surjections by \Cref{asm:sec3}, the morphism $Tf_0 \colon TX_0 \to
      TY_0$ is surjective, so every quotient $e$ of the free
      $\MT$-algebra~$(TY_0, \mu_{Y_0})$ in~$\D^\MT$ yields a quotient 
      $\overline e 
      \defeq
      e \o Tf_0$ of~$(TX_0, \mu_{X_0})$. Due
      to~\Cref{prop:characterisation-hatT}, $\hatt f$ is the mediating morphism
      \[
        \xymatrix{
          \hatt X \ar[r]^{\hatt f} \ar@{->>}[rd]_{\br{\overline e}} &
          \hatt Y \ar@{->>}[d]^{\br{e}} \\
          & \hatJ U(A, \alpha)
        }
      \]
      for all quotients~$e$ of~$(TY_0, \mu_{Y_0})$. Each component 
$\br{\overline e}$ 
      is
      surjective by \Cref{prop:projection-is-surjective}, so the mediating
      morphism~$\hatt f$ is also surjective by
      \Cref{lem:surjections-between-cofiltered-diagrams}.
    \item Now let $X$ and $Y$ be arbitrary. By \Cref{rem:arrowcatclfp} the 
    morphism $f$ is the cofiltered limit in $\hatD^\to$ of all morphisms 
    $(x_i, y_i) \colon f \to f_i$ with $f_i\colon X_i \to Y_i$ in 
          $\D_f$. Let $f_i = e_i \o m_i$ be the canonical 
      factorisation of $f_i$ in $\hatD$.
      \[
        \xymatrix@-1em{
          X \ar@{->>}[rr]^{f} \ar[dd]_{x_i} & & Y \ar[dd]^{y_i} 
          \ar@{-->}[ld]_{h_i} \\
          & X_i' \ar@{ >->}[rd]_{m_i} \\
          X_i \ar@{->>}[ru]^{e_i} \ar[rr]_{f_i} & & Y_i .
        }
      \]
      There exists a unique morphism~$h_i \colon Y \to X_i'$ by the diagonal
      fill-in property. It follows that surjections $(e_i)$ form an initial 
      cofiltered
      subdiagram.  Since $\hatt $ preserves cofiltered limits by (a),
      $\hatt f$ is the limit of  $\hatt e_i \colon \hatt X_i\to \hatt Y_i$.
      However, each $\hatt e_i$ is surjective as proved in (1), so by
      \Cref{lem:surjections-between-cofiltered-diagrams} it follows that 
$\hatt f$ is surjective.\\

\noindent (c) We need to prove that finite $\hatT$-algebras are finitely copresentable in $\hatD^{\hatT}$. Consider the category~$\Alg(\hatt)$ of algebras for the
  \emph{functor}~$\hatt$. In \cite[Lemma
  3.2]{Adamek2004a} it is proved that for a cofinitary functor on a locally 
  finitely copresentable
  category, every algebra $(A, \alpha)$ with $A$ finitely copresentable is
  finitely copresentable in~$\Alg(\hatt)$. Now observe that $\hatD^{\hatT}$ is 
  a full subcategory of~$\Alg(\hatt)$ closed under
  limits, $\hatD$ is
    locally finitely copresentable, and finite objects in $\hatD$ are finitely 
    copresentable, see \Cref{re:hatD-is-procompletion}. Hence, finite 
    $\hatT$-algebras are finitely
  copresentable in $\hatD^\MT$.
\end{enumerate}
\end{proof}

\begin{rem}\label{re:forgetful-monad-morphism}
Every right adjoint preserves right Kan extensions, so $(V\hatt, 
V\epsilon)$
  is the right Kan extension of $V\hatJ U$ along~$\hatJ U$.  Moreover, there 
  is 
  a natural transformation $\gamma: TV\hat J U \to V\hat J U$ whose component 
  at $(A,\alpha)\in \D^\MT_f$ is given by $\alpha$ itself:
  \[
    \gamma_{(A, \alpha)} \defeq \alpha \colon TU(A, \alpha) \to U(A, \alpha).
  \]
 Therefore, by the universal property of $V\hatt$,
  there is a unique natural transformation $\phi\colon TV \to V\hatt $ such
  that~$\gamma = V \epsilon \o \varphi \hatJ U$. By the limit formula for 
  right Kan extensions and Remark 
 \ref{rem:hattconst} 
  the component 
  $\phi_{\hat X}$ for $X\in \D$ is the 
  unique mediating morphism making the following diagram commute for all $h: 
  (TX,\mu_X) \to (A,\alpha)$ with $(A,\alpha)\in\D^\MT_f$:
  \begin{equation}\label{eq:phix}
    \vcenter{
      \xymatrix{
        TV\hat X \ar[r]^{\phi_{\hat X}} \ar[d]_{TV(\widehat{h\eta_{X}})} &
        V\hatt \hat X \ar[d]^{V\br{h}} \\
        TA= TV\hat A \ar[r]_>>>>>>>>{\alpha} & **[r] V\hat A = A.
      }
    }
  \end{equation} 
  In particular, putting $h=\alpha$ for $(A,\alpha)\in\D^\MT_f$, we get the 
  commutative triangle
   \begin{equation}\label{eq:phia}
      \vcenter{
        \xymatrix{
          TV A \ar[r]^{\phi_{ A}} \ar[dr]_\alpha &
          V\hatt  A \ar[d]^{V\br{\alpha}} \\
           & A.
        }
      }
  \end{equation}
%  \begin{equation} 
%    \vcenter{
%      \xymatrix{
%        \D_f^\MT \ar[r]^{U} \ar[rd]_{U}
%        \drtwocell<\omit>{<-2.5>\gamma} &
%        \D \ar[d]^T  \\
%        & \D 
%      }
%    }
%    \quad=\quad
%    \vcenter{
%      \xymatrix{
%        \D_f^\MT \ar[r]^{\hatJ U} \ar[rd]_{\hatJ U} 
%        \drtwocell<\omit>{<-2.5>\epsilon} & \hatD
%        \ar[d]^{\hatt } \drtwocell<\omit>{\phi}
%        \ar[r]^{V} & \D \ar[d]^{T} \\
%        & \hatD \ar[r]_{V} & \D.
%    }
%    }
%  \end{equation}
\end{rem}

\begin{proposition} \label{prop:forgetful-functor}
\begin{enumerate}[(a)] 
\item $\phi\colon TV \to V\hatt$ is a \emph{monad morphism~\cite{Street1972}}, 
  i.e.\
  the following diagrams commute:
    \[
      \vcenter{
        \xymatrix@R-2em{
          & TV \ar[dd]^{\phi} \\
          V \ar[ru]^{\eta V} \ar[rd]_{V\hat \eta} & \\
          & V\hatt
        }
      }
      \quad\text{and}\quad
      \vcenter{
        \xymatrix@-.5em{
          TTV \ar[d]_{\mu V} \ar[r]^{T\phi} & TV\hatt  \ar[r]^{\phi
            \hatt} & V\hatt \hatt  \ar[d]^{V\hat\mu} \\
          TV \ar[rr]_{\phi} & & V\hatt 
        }
      }
    \]
\item For finite $X$ the morphism $\phi_X\colon TVX \to V\hatt X$ is dense.
\end{enumerate}    
\end{proposition}

\begin{proof}
(a)  The following pasting diagrams use the universality of $\epsilon$ and the 
monad
  laws.
  \begin{itemize}
    \item The preservation of unit $\phi \o \eta V = V\hat\eta$:
      \begin{align*}
        &
        \vcenter{
          \xymatrix{
            \D^\MT_f \ar[d] \ar[r] \drtwocell<\omit>{\epsilon} &
            \hatD \dtwocell{\hat\eta} \\
            \D_f \ar[r] & \hatD \ar[r]_V & \D
          }
        }
        \quad=\quad
        \vcenter{
          \xymatrix{
            \D^\MT_f \ar[d] \ar[rr] \drrtwocell<\omit>{\id} & & \D
            \ar[d]^{I} \\ \D_f \ar[rr] &  & \D
          }
        }
        \\
        =
        &
        \vcenter{
          \xymatrix{
            \D^\MT_f \ar[d] \ar[rr] \drrtwocell<\omit>{\gamma} & & \D
            \dtwocell{\eta} \\ \D_f \ar[rr] &  & \D
          }
        }
        \quad=\quad
        \vcenter{
          \xymatrix{
            \D^\MT_f \ar[d] \ar[r] \drtwocell<\omit>{\epsilon} &
            \hatD \ar[d] \ar[r]^V \drtwocell<\omit>{\phi}
            & \D \dtwocell{\eta} \\
            \D_f \ar[r] & \hatD \ar[r]_V & \D
          }
        }
      \end{align*}
    \item The preservation of multiplication $\phi \o \mu_V = V\hat\mu \o
      \phi\hatt \o T\phi$: 
      \begin{align*}
        &
        \vcenter{
          \xymatrix{
            \D^\MT_f \ar[d]_{U_f} \ar[r] \drtwocell<\omit>{\epsilon} &
            \hatD \ar[d]^{\widehat{T}} \ar[r]^V \drtwocell<\omit>{\phi}
            & \D \dtwocell{\mu} \\
            \D_f \ar[r] & \hatD \ar[r]_V & \D
          }
        }
         = 
        \vcenter{
          \xymatrix{
            \D^\MT_f \ar[d] \ar[rr] \drrtwocell<\omit>{\gamma} &
            & \D \dtwocell{\mu} \\
            \D_f \ar[rr] &  & \D
          }
        }
        \\
        = &
        \vcenter{
          \xymatrix@C+2em@R-.5em{
            \drtwocell<\omit>{\gamma} & \D \ar[d]^T \\
            \D_f^\MT \ar[r] \drtwocell<\omit>{\gamma}
            \ar@/^1pc/[ru]\ar@/_1pc/[rd] & \D \ar[d]^T \\
            & \D
          }
        }
         =  
        \vcenter{
          \xymatrix@C+2em@R-.5em{
            \drtwocell<\omit>{\epsilon} & \D \dtwocell{\phi} \\
            \D_f^\MT \ar[r] \drtwocell<\omit>{\epsilon}
            \ar@/^1pc/[ru]\ar@/_1pc/[rd] & \D \dtwocell{\phi} \\
            & \D
          }
        }
         =
        \vcenter{
          \xymatrix@C+2em@R-.5em{
             \ddrtwocell<\omit>{\epsilon}    & \D \dtwocell{\phi}
             \ddlowertwocell{\mu}  \\
            \D_f^\MT 
            \ar@/^1pc/[ru]\ar@/_1pc/[rd] & \D \dtwocell{\phi} \\
            & \D
          }
        }
      \end{align*}
  \end{itemize}

\noindent (b) If $X$ is finite we have for all $h: (TX,\mu_X) \to (A,\alpha)$ 
with $(A,\alpha)\in \D^\MT_f$:
 \[ \alpha\o TV(\widehat{h\eta_X}) = \alpha \o Th\o T\eta_X = h\o \mu_X \o 
 T\eta_X = h  \]
 using that $h$ is a $\MT$-homomorphism and the unit law of the monad $\MT$.
 In particular, if $h$ is surjective, so is the cone $\alpha\o 
 TV(\widehat{h\eta_X})$ in \eqref{eq:phiv}. By 
 \Cref{prop:characterisation-hatT} and Lemma \ref{lem:dense-lemma} 
 this implies that $\phi_X$ is dense.
\end{proof}

\begin{proof}[\Cref{prop:finite-algebras}]
(a) The maps $(A, \alpha) \mapsto ({A}, \br{\alpha})$ and $h
  \mapsto \widehat{h}$ define a functor 
  \begin{equation} 
    \label{eq:comparison-functor}
    K^\MT \colon \D_f^\MT \to \hatD_f^{\hatT}.
  \end{equation}
To see this we only need to prove that for every $\MT$-homomorphism $h: 
(A,\alpha)\to (B,\beta)$ is also a $\hatT$-homomorphism $(A,\alpha^+)\to 
(B,\beta^+)$. Indeed, in the diagram below the right hand square commutes 
when precomposed with $\hat\eta_A$, so it commutes by the universality of 
$\hat\eta_A$.
\[
\xymatrix{
A \ar[r]^{\hat\eta_A} \ar[d]_h \ar@/^2em/[rr]^\id & \hatt A \ar[d]^{\hatt h} 
\ar[r]^{\alpha^+} & A \ar[d]^h\\
B \ar[r]_{\hat\eta_B} \ar@/_2em/[rr]_\id & \hatt B \ar[r]_{\beta^+} & B
}
\]

%Indeed, for the universal natural transformation $\epsilon\colon \hatt \hatJ 
%U 
%\to \hatJ
%U$ of the right Kan extension $\Ran_{\hatJ U}\hatJ U = \hatt$ we have 
%\[
%  \epsilon_{(A, \alpha)} = \id_A^*: \hatt  A \to A
%\]
%for each $(A, \alpha) \in \D_f^\MT$, see 
%\Cref{not:limitcone} and Theorem X.3.1 of~\cite{MacLane1978}.
%But $\id_A^* = \br{\alpha}$, because
%the identity $\id_{{A}} \colon {A} \to \hatJ U(A, \alpha)$ corresponds 
%to
%$\id_A \colon A \to U(A, \alpha)$ by the left adjoint $\hat{\,\cdot\,}\colon 
%\D \to \hatD$
%and to $\alpha\colon (TA, \mu_A) \to (A, \alpha)$ by the right adjoint
%$U^\MT\colon \D^\MT \to \D$. The naturality entails that morphisms $f\colon 
%(A, \alpha) \to (B, \beta)$
%between finite $\MT$-algebras are morphisms $\hatJ f\colon K^\MT(A, \alpha)
%\to K^\MT(B, \beta)$ of finite $\hatT$-algebras. 

\noindent(b) Conversely, the monad morphism $\phi: TV\to V\hatt$, see 
\Cref{re:forgetful-monad-morphism} and
\Cref{prop:forgetful-functor}, induces a
functor 
\[V^\MT \colon
    \hatD^{\hatT}_f \to \D^\MT_f\] 
    mapping $(A, \alpha)$ to~$(VA, V\alpha \o \phi_A)$ and $h\colon (A, 
    \alpha) \to (B, \beta)$ to $Vh$.

%In \Cref{sec:profinite-monad}, we exhibit that for $X$ finite, $\hatt X$ 
%is a
%limit with the cone $\br{h}$ indexed by 
%$\MT$-algebra homomorphisms $h\colon (TVX, \mu_{VX}) \to (A, \alpha)$.  This
%limit is preserved by the right adjoint $V$, so there exists a unique
%morphism~$\phi'_X\colon TVX \to V\hatt X$ with~$h = V\br{h} \o
%\phi'_X$.
%However, $h$ is of the form~$\alpha\o F^\MT h$ by adjunction. Putting these
%together, we have
%\[
%  \vcenter{
%    \xymatrix{
%      TVX \ar[r]^{\phi'_X} \ar[d]_{F^\MT h = Th} &
%      V\hatt X \ar[d]^{V\br{h}} \\
%      TA \ar[r]_{\alpha} & **[r] A = V\hat{A}.
%    }
%  }
%\]
%The lower-left corner is exactly the cone deriving $\phi$ introduced in
%\Cref{prop:forgetful-functor} in terms of the limit formula of right Kan
%extension, so $\phi_X = \phi'_X$.
%
%With \Cref{prop:characterisation-hatT}, this observation leads to the 
%following
%identification between finite $\MT$-algebras and finite $\hatT$-algebras:
%\begin{proposition}
%  The forgetful functor $V^\MT \colon \hatD^{\hatT} \to \D^\MT$ restricted to
%  finite $\hatT$-algebras is an isomorphism with $K^\MT$ defined
%  in~\eqref{eq:comparison-functor} as its inverse.
%\end{proposition}
%\begin{proof}

\noindent (c) It remains to prove that $K^\MT$ and $V^\MT$ are mutually 
inverse 
functors. Clearly this holds on morphisms, since both functors are identity on 
morphisms. As for objects, for every
  finite $\MT$-algebra $(A, \alpha)$ we have
  \begin{align*}
    V^\MT K^\MT (A, \alpha)
    & = V^\MT (\hat{A}, \br{\alpha}) 
    && \{\,\text{by definition}\,\} \\
    & = (V\hat{A}, V\alpha^+ \o \phi_A) 
    && \{\, \text{by definition}\,\} \\
    & = (A, \alpha).
    && \{\, \text{by \eqref{eq:phia}}\,\}
  \end{align*}

  Conversely, let $(A, \alpha)$ be a finite $\hatT$-algebra
  and $(A, \alpha') \defeq K^\MT V^\MT (A, \alpha)$. Applying~$V^\MT$ on both
  sides, we have $V^\MT(A, \alpha') = V^\MT K^\MT V^\MT (A, \alpha) = V^\MT(A,
  \alpha)$, using that $V^\MT K^\MT = \Id$ as proved above. That is, 
  $V\alpha'$ and 
  $V\alpha$ agree on the image 
  of~$\varphi_A$,
  i.e.\ the diagram
  \[
    \xymatrix{
      TVA \ar[r]^{\varphi_A} & V\hatt A \ar@<-1ex>[r]_{V\alpha}
      \ar@<1ex>[r]^{V\alpha'} & VA
    }
  \]
  commutes. By \Cref{prop:forgetful-functor}(b) the morphism $\phi_A$ is 
  dense, which implies $V\alpha' = V\alpha$ since $A$ is a 
  Hausdorff space. Since the forgetful 
  functor $V$ is faithful, we conclude $\alpha = \alpha'$. 
\end{proof}

\begin{proof}[\Cref{lem:eqtovar}]
Let us first consider the unordered case. Since intersections of 
pseudovarieties are pseudovarieties, we may assume that 
our class is presented by a single equation $u=v$ over a finite set $X$.
\begin{enumerate}[(a)]
\item Closure under finite products: clearly the trivial one-element 
$\MT$-algebra satisfies all profinite equations. 
Let $(A_0,\alpha_0)$ and $(A_1,\alpha_1)$ be finite $\MT$-algebras satisfying 
$u=v$ and let $p_i\colon (A_0\times A_1,\alpha)\to (A_i,\alpha_i)$ be their 
product.  For 
any  $\MT$-homomorphism $h\colon T\psi_X\to A_0\times A_1$ put $h_i = p_i \o 
h$. Then
\[ p_i \o \br{h}  (u) = \br{h_i}  (u) =  \br{h_i} (v) = 
p_i \o\br{h}  (v) \] so $\br{h}  (u) = \br{h}  (v) $  since 
the projections $p_i$ are jointly monic. We conclude that 
$(A_0\times A_1,\alpha)$ satisfies $u=v$.
\item Closure under subalgebras: let $(A,\alpha)$ be a finite $\MT$-algebra 
satisfying $u=v$ and $m\colon (B,\beta)\monoto (A,\alpha)$ be a subalgebra. 
Then 
for any $\MT$-homomorphism $h\colon T\psi_X \to B$ we have
\[ m\o \br{h} (u) = \br{(m\o h)}  (u) =  \br{(m\o h)}  (v) =  m\o 
\br{h} (v) \]
so $\br{h}(u) = \br{h} (v)$ since $m$ is monic. Hence $(B,\beta)$ 
satisfies 
$u=v$.
\item Closure under quotients: let $(A,\alpha)$  be a finite $\MT$-algebra  
satisfying $u=v$ and $e\colon 
(A,\alpha)\epito (B,\beta)$ be a quotient. For any $h_0: \psi_X\to B$ choose a 
morphism $h_0': \psi_X\to A$ with $h_0 = e\o h_0'$, cf. \Cref{rem:homtheorem}. 
The corresponding $\MT$-homomorphisms $h: T\psi_X\to B$ and $h': T\psi_X\to A$ 
satisfy $h = e\o h'$, so
\[ \br{h}  (u) = e\o \br{(h')} (u) = e\o \br{(h')}(v) = 
\br{h}  (v), \]
proving that $(B,\beta)$ satisfies $u=v$.
\end{enumerate}
The ordered case in analogous: replace profinite equations by inequations, in 
(a) use that the projections $p_i$ are jointly order-reflecting, and in (b) 
use that $m$ is order-reflecting.
\end{proof}

\begin{proof}[\Cref{lem:keylemma}]
Consider first the unordered case.
For any $\MT$-homomorphism $h\colon TX\to A$ with $(A,\alpha)\in \V$ we have 
the 
commutative triangle \eqref{eq:phiv}.
Since the projections $\brv{h}$ are jointly monic, it follows that 
$\phi^\V_\hatX u = \phi^\V_\hatX v$ iff $\br{h} u = \br{h} v$ 
for all 
$h$, i.e.\ iff every $(A,\alpha)\in \V$ satisfies $u=v$. For the ordered case 
replace equations by inequations and use that the 
projections $\brv{h}$ are jointly order-reflecting.
\end{proof}

\begin{proof}[Theorem \ref{cor:reitermanquasi}]
3$\Ra$1 requires a routine verification analogous to the proof of
\Cref{lem:eqtovar}, and 1$\Ra$2 is \Cref{thm:reiterman}. For 2$\Ra$3 
consider first the unordered case. We may assume that $\V$ is presented by a 
single profinite equation $u=v$ with $u,v$ elements of some $X\in\D_f$. 
Express $X$ as a quotient $q\colon \Phi_Y\epito X$ for some finite set $Y$. 
Let 
$\{\,(u_i,v_i) : i \in I\,\}$ be the kernel of $\hat q\colon  \hat{\Phi_Y} 
\epito 
X$ (consisting of all pairs $(u_i,v_i)\in \hat{\Phi_Y}\times\hat{\Phi_Y}$ with 
$\hat q (u_i) = \hat q (v_i)$), and choose $u',v'\in \hat{\Phi_Y}$ with $\hat 
q(u') = u$ and  $\hat 
q(v') = v$.  We claim that a finite object $A\in\D_f$ satisfies the 
profinite equation $u=v$ iff it satisfies the profinite implication 
\begin{equation}
\bigwedge_{i\in I} u_i=v_i ~\Ra~ u'=v',\tag{\ref{eq:imp}}
\end{equation}
which proves that $\V$ is presented by that implication. 

For the ``if'' direction suppose that $A$ satisfies \eqref{eq:imp}. For any morphism $h\colon X\ra A$ in $\D$ we have the commutative triangle below:
\[
\xymatrix{
\hat{\Phi_Y} \ar[r]^>>>>>>{\hat q} \ar[d]_{\widehat{hq}} & \hatX=X 
\ar[dl]^{\hat{h}=h}\\
A & 
}
\]
The morphism $\widehat{hq}$ merges $u_i,v_i$ for all $i$ since $\hat q$ does, 
so by \eqref{eq:imp} it also
merges 
$u',v'$. We conclude
\[ \hat{h}(u) = \hat{h}\hat q(u') = \widehat{hq}(u') = \widehat{hq}(v')  =
\hat{h}\hat q(v') = \hat{h}(v),\]
so $A$ satisfies $u=v$.

For the ``only if'' direction, suppose $A$ satisfies $u=v$, and let $h\colon
\Phi_Y\ra A$ be a morphism with $\hat{h}(u_i)=\hat{h}(v_i)$ for all $i$. Since
$(u_i,v_i)$ are precisely the pairs merged by $\hat q$, the homomorphism
theorem (see \Cref{rem:homtheorem}) yields a morphism $h'\colon X \to A$ 
with $h'\o \hat q = \hat{h}$. Since 
$A$ satisfies $u=v$, it follows that $h'(u)=h'(v)$ and hence 
 \[ \hat{h}(u') = h'\o  \hat q (u') = h' (u) = h'(v) = h'\o  \hat q (v') = 
 \hat{h}(v').\]
Hence $A$ satisfies the implication \eqref{eq:imp}.

The ordered case is analogous to the above argument, replacing 
equations by inequations and using the homomorphism theorem for ordered 
algebras in the ``only if'' direction.
\end{proof}

\end{document}